\documentclass{article}

\usepackage{microtype}
\usepackage{graphicx}
\usepackage{subfigure}
\usepackage{booktabs} 

\usepackage{hyperref}

\usepackage[main, final]{neurips_2026}

\usepackage{amsmath}
\usepackage{amssymb}
\usepackage{mathtools}
\usepackage{amsthm}
\usepackage{soul}

\usepackage[capitalize,noabbrev]{cleveref}

\usepackage{xcolor}
\definecolor{afiablue}{RGB}{61,159,207}
\definecolor{afiared}{RGB}{167,75,68}
\definecolor{afialightblue}{RGB}{158,193,232}

\usepackage{xcolor}         

\usepackage[utf8]{inputenc} 
\usepackage[T1]{fontenc}    
\usepackage{hyperref}       
\usepackage{url}            
\usepackage{booktabs}       
\usepackage{amsfonts}       
\usepackage{nicefrac}       
\usepackage{microtype}      
\usepackage{lipsum}
\usepackage{enumitem}
\usepackage{bbm}
\usepackage{graphicx}
\graphicspath{ {./images/} }
\usepackage{natbib}
\usepackage{adjustbox}

\usepackage{amsmath}
\usepackage{amssymb}
\usepackage{mathtools}
\usepackage{centernot}
\usepackage{amsthm}

\usepackage{tikz}
\usetikzlibrary{positioning, arrows.meta}
\usepackage{tikz-cd}
\usetikzlibrary{shapes.geometric,arrows.meta,positioning}

    \theoremstyle{plain}
    \newtheorem{theorem}{Theorem}[section]
    \newtheorem{proposition}{Proposition}[section]
    \newtheorem{corollary}{Corollary}[section]
    \newtheorem{lemma}{Lemma}[section]
    \theoremstyle{definition}
    \newtheorem{definition}[theorem]{Definition}
    
    \theoremstyle{remark}
    \newtheorem{remark}[theorem]{Remark}
    \usepackage{stmaryrd}
    \usepackage{xparse} \DeclarePairedDelimiterX{\Iintv}[1]{\llbracket}{\rrbracket}{\iintvargs{#1}}
    \NewDocumentCommand{\iintvargs}{>{\SplitArgument{1}{,}}m}
    {\iintvargsaux#1} %
    \NewDocumentCommand{\iintvargsaux}{mm} {#1\mkern1.5mu,\mkern1.5mu#2}

    \DeclareMathOperator{\Cauchy}{Cauchy}

    \DeclareMathOperator{\Var}{Var}

    \DeclareMathOperator{\rank}{rank}

\usepackage{thmtools}
\usepackage{thm-restate}
    
\usepackage[textsize=tiny]{todonotes}
\usepackage{xcolor}

\title{Privacy Amplification Persists under Unlimited Synthetic Data Release}

\author{
Clément Pierquin$^{1, 3}$\thanks{Corresponding author at \texttt{clement.pierquin@craft-ai.fr}} \quad Aurélien Bellet$^2$ \quad Marc Tommasi$^1$ \quad Matthieu Boussard$^3$\\
$^1$Université de Lille, Inria,
CNRS, Centrale Lille, UMR 9189 CRIStAL,
F-59000 Lille, France\\
$^2$ PreMeDICaL team, Inria,  Idesp, Inserm, Université de Montpellier\\
$^3$ Craft AI, Paris, France\\}
\begin{document}

\maketitle
\begin{abstract}
  We study \emph{privacy amplification by synthetic data release}, a phenomenon in which differential privacy guarantees improve when releasing only synthetic data rather than the private generative model itself. Recent work by~\cite{pierquin2025} established the first formal amplification guarantees for a linear generator, but they apply only in asymptotic regimes where the model dimension far exceeds the number of released synthetic records, limiting practical relevance. In this work, we show that this restriction is not fundamental and establish that, surprisingly, for linear generators, \emph{privacy amplification can persist even when releasing an unbounded number of synthetic records}. To prove this result, we develop new analytical tools for characterizing the privacy loss induced by synthetic data release. In particular, we leverage a sufficient-statistics reduction to characterize privacy leakage, and introduce a novel criterion that upper bounds Rényi divergences in terms of Fisher information. Our analysis provides structural insights that may guide the development of tighter privacy guarantees for more complex release mechanisms. Finally, experiments with variational autoencoders trained using DP-SGD further support the relevance of our theory beyond the linear setting.
\end{abstract}

\section{Introduction}

Differential privacy (DP)~\citep{Dwork2014} provides a rigorous framework for protecting sensitive data in machine learning. Privacy-preserving models can be obtained through methods like output perturbation~\citep{chaudhuri2011differentiallyprivateempiricalrisk,Zhang20172,Lowy2024}, which adds noise to a trained model, or noisy gradient descent algorithms, which inject noise during training~\citep{private_sgd,Bassily2014a,Abadi2016,Feldman2018}. 
A key property of DP is \emph{post-processing}: once a model is differentially private, any further processing or analysis preserves its guarantees. This makes \emph{differentially private generative models} \citep{Zhang2017,xie2018,mckenna2019,Jordon2018,McKenna2021,Lee2022,dockhorn2022,bie2023private} particularly appealing, as the synthetic data they generate automatically inherits the model's privacy guarantees, allowing synthetic datasets to be shared safely for downstream use, in contrast to heuristic approaches that may not provide meaningful privacy protection~\citep{zhao2025does,ponomareva2025dpfydatapracticalguide}.

Empirical evidence, however, suggests that synthetic data can sometimes provide stronger privacy than the formal guarantees of the underlying generative model~\citep{Annamalai2024, Houssiau2022}, hinting at a possible privacy amplification effect. One intuitive explanation is that privacy leakage may be lower when the number of synthetic samples is small relative to the complexity of the generative model.

~\cite{pierquin2025} initiated the formal study of this phenomenon in linear regression under output perturbation. Their analysis attributes the amplification effect to the latent randomness in the data generation process. They showed that releasing a limited number of synthetic points can improve privacy guarantees over direct model release. However, these guarantees only hold in a highly restrictive regime where the number $n_{\mathrm{syn}}$ of synthetic records is small compared to the model dimension $d$.  This yields bounds that are vacuous in practical release regimes (see Figure~\ref{fig:summary-heatmap}). It also leaves open a qualitatively different question: \emph{can amplification persist when the curator releases many, or even infinitely many, synthetic samples?}

In this paper, we answer this question positively and substantially improve upon prior results by developing a unified framework for privacy amplification in synthetic data released by a linear generator. Rather than analyzing a particular regression objective, we analyze an abstraction that contains the setting of~\cite{pierquin2025} as a special case, while making it possible to identify the minimal structural source of amplification.
Our analysis identifies and precisely characterizes the mechanisms driving privacy amplification in this tractable model, offering insights that may guide the study of more general synthetic data release scenarios.
Specifically, our main contributions are as follows:\looseness=-1
\begin{itemize}[itemsep=1pt, parsep=2pt, topsep=1pt]
\item We show that releasing \textbf{infinitely many synthetic records $(n_{\mathrm{syn}}\to+\infty)$ can still yield significant privacy amplification}—a surprising result that depends on parameter norm instead of $n_{\mathrm{syn}}$.
We propose a unified characterization of privacy in the high-privacy regime,
and also derive a general privacy amplification upper bound.
\item To establish these results, we develop new technical tools based on Fisher information. Specifically, we introduce a criterion that allows controlling Rényi divergences via Fisher information, and we derive upper bounds on the Fisher information for the non-centrality parameter of non-central chi-squared and non-central Wishart distributions. These results are more general and unified than existing analyses, and may be of independent interest to the statistics community.
\item To complement our theoretical results, we perform a series of privacy auditing experiments on variational autoencoders trained with DP-SGD. The measured privacy leakage supports the relevance of our theory beyond the linear setting. We also numerically estimate the Rényi DP guarantees and compare them with our theoretical bounds across different parameters. Strikingly, we find that the privacy guarantees converge rapidly to their $n_{\mathrm{syn}} \to +\infty$ limits, showing that focusing on this asymptotic regime captures the key amplification effect and provides useful bounds for practical amounts of synthetic data.

\end{itemize}

All proofs and additional discussions are provided in the Appendix.

Before presenting our main results in more detail, we highlight an unexpected and conceptually important insight that emerged from our analysis. \textbf{We invalidate the natural intuition}, put forward by \citet{pierquin2025} and~\cite{ge2026}, \textbf{that even in simple linear models, releasing an unbounded number of synthetic records would systematically eliminate any privacy benefit}. According to this intuition, an adversary could, in principle, exploit central limit effects to detect even the smallest distributional shifts and reconstruct the private generative model with arbitrary precision.
Our analysis shows this is not the case: the privacy loss can be effectively upper-bounded without restricting the number of released synthetic records. This yields a privacy accounting rule that, even for arbitrarily large numbers of synthetic records, is strictly tighter than the standard post-processing privacy guarantee obtained when releasing the model parameters.

\section{Setting and Overview of our Results}
\label{sec:l_infty}

\textbf{Rényi Differential Privacy.} 
In this work, we express privacy guarantees using Rényi Differential Privacy (RDP), a variant of differential privacy that quantifies privacy in terms of the Rényi divergence \citep{Mironov2017}. 
Two datasets $\mathcal{D}$ and $\mathcal{D}'$ of fixed size $n$ are called adjacent if they differ in a single data point~\citep{Mironov2017}. Then, a mechanism $\mathcal{M}$ satisfies $(\alpha, \varepsilon)$-RDP if, for any adjacent datasets $\mathcal{D}$ and $\mathcal{D}'$,
\[D_\alpha(\mathcal{M}(\mathcal{D}) \| \mathcal{M}(\mathcal{D}')) \leq \varepsilon,\]
with $D_\alpha(P,Q) = \frac{1}{\alpha-1} \log \mathbb{E}_{x \sim Q}[(P(x)/Q(x))^\alpha]$ the Rényi divergence of order $\alpha>1$ between distributions $P$ and $Q$. By a slight abuse of notation, if $V \sim P$ and $W \sim Q$, we write $D_\alpha(V,W) = D_\alpha(P,Q)$.  

\textbf{The setting: synthetic data and linear generators.} Many approaches for private synthetic data release are modeled as a two-stage randomized mechanism \citep{xie2018,Jordon2018,Lee2022,dockhorn2022}. In the first stage, a differentially private procedure $\mathcal{M}$ outputs the parameters $V=\mathcal{M}(\mathcal{D})$ of a generative model $f_V$ trained on a sensitive dataset $\mathcal{D}$. In the second stage, the learned generator $f_V$ draws a latent noise $Z = (Z_1, \dots Z_{n_{\text{syn}}})$ from some (typically Gaussian) distribution $\mu^{\otimes {n_{\text{syn}}}}$ and outputs $n_{\text{syn}}$ synthetic records $f_V(Z) = (f_V(Z_1),\dots, f_V(Z_{n_{\text{syn}}}))$. Therefore, the privacy loss associated to synthetic data release is the Rényi divergence $D_\alpha(f_V(Z), f_W(Z))$, which is upper-bounded by  $D_\alpha(V,W)$ by the post-processing property of DP. For adjacent datasets $\mathcal{D}, \mathcal{D}'$, privacy amplification by synthetic data release then corresponds to establishing a stronger bound of: $D_\alpha(f_V(Z), f_W(Z)) \leq \eta_{\mathcal{D},\mathcal{D}'} D_\alpha(V,W)$, with $\eta_{\mathcal{D},\mathcal{D}'} < 1$, which can be attributed to the additional randomness introduced by $Z$.

For realistic generators $f_V$ (e.g. deep generative models), the divergence $D_\alpha(f_V(Z), f_W(Z))$ is intractable. We therefore study a tractable proxy that extends the setting of~\citet{pierquin2025}:
\begin{enumerate}
    \item We model $f_V$ as a linear function $f_V(Z) = ZV$, where $Z \in \mathbb{R}^{n_{\text{syn}} \times d}$ and $V \in \mathbb{R}^{d \times k}$ with $n_{\text{syn}}$ the number of synthetic records, $k$ the dimension of synthetic records and $d$ the width of the model.
    Similar idealized models have been used to study model collapse in generative AI~\cite{Dohmatob2024, Gerstgrasser2024} and learning dynamics of generative models~\citep{Lucas2019, Gemp2018}.
    \item We consider deterministic parameters $v,w \in \mathbb{R}^{d \times k}$ obtained from adjacent datasets $\mathcal{D}$, $\mathcal{D'}$ via a bounded-sensitivity procedure, which are then privatized with output perturbation, i.e., by adding Gaussian noise.
\end{enumerate}

The bounded-sensitivity assumption is natural, and essentially without loss in the linear generator setting. For instance, under standard boundedness conditions, $\ell_2$-regularized multi-output linear regression has bounded sensitivity \citep{chaudhuri2011differentiallyprivateempiricalrisk,pierquin2025}, meaning that $\|v-w\|_F \leq \Delta$ for some sensitivity $\Delta > 0$ that depends on the problem parameters.

Formally, we consider the following synthetic data release mechanism.
\begin{definition}[Linear generation from Gaussian inputs]
    Let $v \in \mathbb{R}^{d \times k}$ be a deterministic parameter. Let $\mathcal{M}(v) := v + \sigma N$, where $N\in\mathbb{R}^{d\times k}$ has i.i.d. standard normal entries. Independently, draw a matrix $Z\in\mathbb{R}^{n_{\mathrm{syn}}\times d}$ with i.i.d. entries $Z_{ij}\sim\mathcal{N}(0,\sigma_z^2)$. The released synthetic dataset is $Z\mathcal{M}(v)$.
\end{definition}
This definition is strictly more general than the setting of~\cite{pierquin2025}, which only considers $l_2$-regularized linear regression. By the scale-invariance of Gaussian matrices, $D_\alpha(Z(v+\sigma N),Z(w + \sigma N)) = D_\alpha(Z(\frac{1}{\sigma} v + N), Z(\frac{1}{\sigma}w + N))$. Thus, the relevant quantity are the ratios $\frac{1}{\sigma}v$ and $\frac{1}{\sigma}w$. Therefore, without loss of generality, we assume $\sigma_z = \sigma = 1$ throughout the paper.

 We let the random variables $V = \mathcal{M}(v)$ and $W = \mathcal{M}(w)$, for $v,w \in \mathbb{R}^{d \times k}$ obtained from adjacent datasets. We formally define the notion of \emph{privacy amplification by synthetic data release} as follows.
 \begin{definition}[Privacy Amplification by Synthetic Data Release]
 For adjacent datasets $\mathcal{D}, \mathcal{D'}$ yielding parameters $v$ and $w$, we define privacy amplification as:
 \[D_\alpha(f_V(Z), f_W(Z)) \leq\eta_{\mathcal{D},\mathcal{D}'} (n_{\mathrm{syn}}, d, k,  v,w) D_\alpha(V,W), \qquad \eta_{\mathcal{D},\mathcal{D}'} (n_{\mathrm{syn}}, d, k,  v,w)< 1.\]
     In this setting, only the synthetic data $f_V(Z) = ZV$ is released. The model parameters $V$ and the latent noise $Z$ remain hidden from the adversary. This condition is necessary for amplification to happen in linear generators, as shown by~\cite{pierquin2025}. Throughout the paper, we distinguish between an amplification factor $\eta_{\mathcal{D}, \mathcal{D}'}$ for a specific pair of datasets $(\mathcal{D}, \mathcal{D}')$, and a uniform RDP guarantee. $\eta_{\mathcal{D}, \mathcal{D}'}$ captures the point of view of an adversary who performs a statistical test to distinguish between $\mathcal{D}$ and $\mathcal{D}'$ and represents the contraction of the privacy loss.
    This translates to uniform RDP guarantees upper bounds whenever $\sup_{\mathcal{D} \sim \mathcal{D}'} \eta_{\mathcal{D},\mathcal{D}'} < 1$. We will discuss both types of results in the paper.
    For brevity, we write $\eta := \eta_{\mathcal{D},\mathcal{D}'}(n_{\mathrm{syn}}, d, k, v, w)$ when it is clear from context.
\end{definition}

\textbf{Summary of our main result.} Our results show that privacy amplification by synthetic data release is not intrinsically a small-release phenomenon (i.e., small $n_{\mathrm{syn}}$). We derive bounds on $D_\alpha(ZV,ZW)$, the privacy loss of releasing synthetic data, and compare them to $D_\alpha(V,W)$, the privacy loss of releasing the generative model parameters. The guarantees of~\cite{pierquin2025} are only relevant for small-releases and become uninformative when $n_{\mathrm{syn}} \gtrsim d/k^2$; furthermore, the constants and Gaussian approximation error make the certificate loose in moderate-dimensional regimes.
Figure~\ref{fig:summary-rates} summarizes our results and compare them to the scaling provided by the prior work of~\citet{pierquin2025}, omitting potentially large absolute constants for readability.\footnote{The results of~\citet{pierquin2025} were originally reported in the $f$-DP framework~\citep{Dong2019}. 
$f$-DP to RDP conversion is tight~\cite{Dong2019}. In particular, any RDP guarantee obtained from~\cite{pierquin2025} is at least as large as the displayed $\eta_{\mathrm{old}}$. Details are provided in Appendix~\ref{app:comparison-prior-work}.} We obtain a unified upper bound capturing both amplification and non-amplification regimes. Strikingly, our upper bounds do not depend on the output dimension nor on the number of synthetic data points $n_{\text{syn}}$, remaining meaningful even when $n_{\mathrm{syn}} \to \infty$. Instead, the amplification factor is  governed by the parameters $v$ and $w$. Moreover, the improvement is not restricted to large-release regimes: Figure~\ref{fig:summary-heatmap} shows, for example, a $5 \times$ improvements over previous bounds already at $k=d=10$ and $n_{\mathrm{syn}} = 1$ ($\|v\|_F,\|w\|_F = 1$). We also show that the dependence on $\max\{\|v\|_F,\|w\|_F\}$ is unavoidable, and discuss regimes where amplification persists indefinitely.

\begin{figure}[htbp]
\centering
\begin{minipage}[bc]{0.4\textwidth}
\centering
\begin{tikzpicture}
{\small
    \node[
      draw,
      fill=purple!10,
      thick,
      rounded corners, 
      inner xsep=12pt,      
      inner ysep=8pt,        
      text width=0.9\textwidth,
      align=left             
    ] (mybox) {
      \textbf{Prior work~\citep{pierquin2025}:} 
     $$\eta_{\rm old}\lesssim \min\left\{1, k\sqrt{\frac{n_{\text{syn}}}{d-k}}\right\}.$$
     \textbf{Our work:}
     $$\eta_{\rm ours} \lesssim \frac{2\max\{\|v\|_F, \|w\|_F\}^2}{2\max\{\|v\|_F, \|w\|_F\}^2 + d - 1}  .$$
    };
    }
  \end{tikzpicture}
  \caption{Summary of our results highlighting the improvements over prior work by \citet{pierquin2025}.}
  \label{fig:summary-rates}
\end{minipage}
\hfill
\begin{minipage}[bc]{0.56\textwidth}
\includegraphics[width=\linewidth]{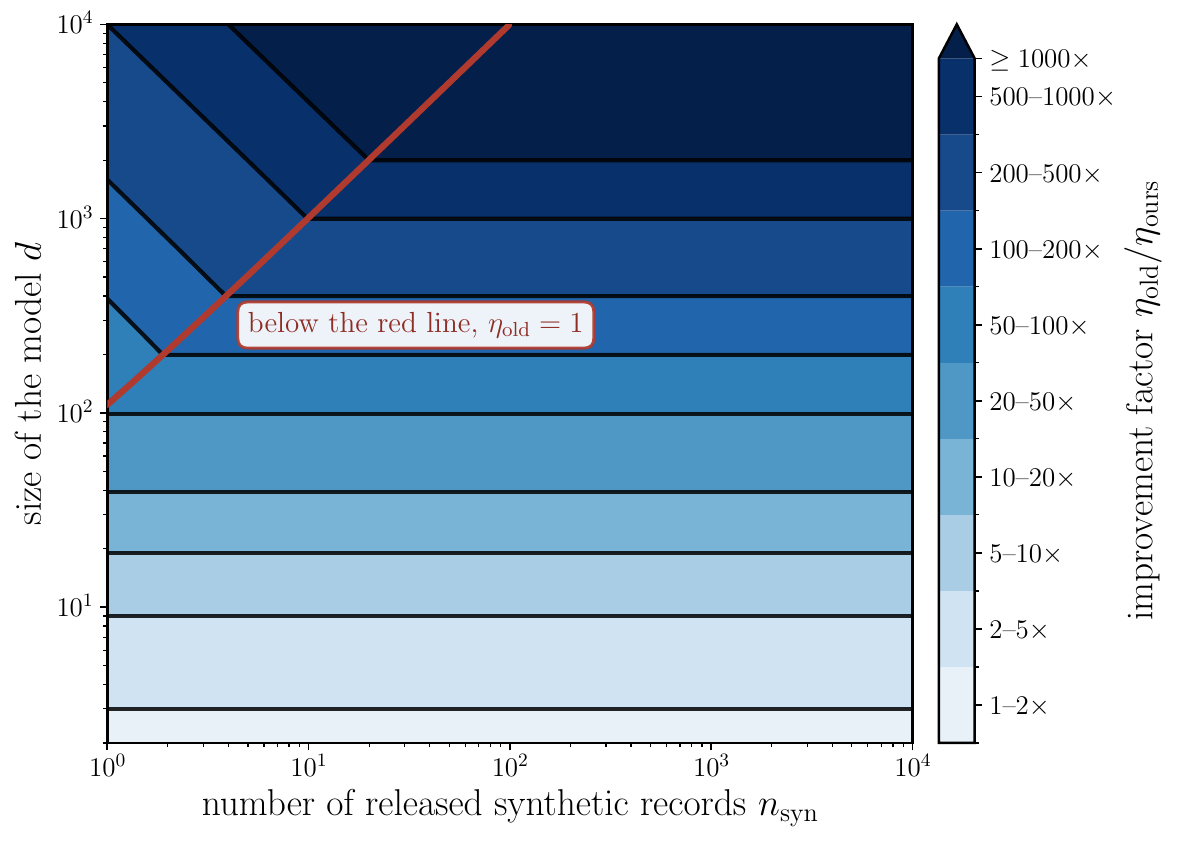}
\caption{Factor of improvement $\eta_{\mathrm{old}}/\eta_{\mathrm{ours}}$ as a function of $d$ and $n_{\mathrm{syn}}$ for $k = 10, \|v\|_F = \|w\|_F = 1$. Below the dark red line, $\eta_{\mathrm{old}} = 1$, meaning that the analysis of~\cite{pierquin2025} does not provide any amplification.}
\label{fig:summary-heatmap}
\end{minipage}

\end{figure}

\section{Relationship between Fisher information and Rényi divergences}
\label{sec:fisher-renyi}
This section presents some key technical results (some of which novel) that we will use to address our main goal:
upper bounding the privacy loss induced in the infinite synthetic data release regime ($n_{\text{syn}}\rightarrow+\infty$). Let $v,w \in \mathbb{R}^{d\times k}$ be the parameter produced by two adjacent datasets, with $\|v-w\|_F \leq \Delta$ the sensitivity, and $P_v$ and $P_w$ the corresponding infinite release distributions of $ZV$ and $ZW$. Unfortunately, as we will see in Section~\ref{sec:privacy-amp-linear}, the Rényi divergence $D_\alpha(P_v, P_w)$ is generally intractable and does not admit simple bounds.

Instead of attempting a direct computation, we take an indirect approach. We embed the parameters $v$ and $w$ into a smooth parametric family $\{P_{\theta}, \theta \in \Theta\}$ that interpolate between $P_v$ and $P_w$, and we study the Fisher information of this family with respect to $\theta$. While still not directly computable, this quantity is typically more tractable to bound than the Rényi divergence itself. Then, we leverage a relationship between Rényi divergences and the Fisher information. Rényi divergences admit a local second-order expansion governed by Fisher information.

\begin{proposition}[Local relationship between Rényi divergences and Fisher information; \citealt{Haussler1997, vanErven2014}]
\label{prop:local-fisher}
    Let $\alpha > 1$. Let $P = \{P_\theta,\; \theta\in \Theta \subset \mathbb{R}\}$ be a regular family of probability distributions (in the sense of Definition~\ref{defi:regularity}). 
    For $\theta \in \Theta$, let the Fisher information $I(\theta) = \mathbb{E}_{X \sim P_\theta}[(\partial_{\theta'} (dP_{\theta'}/dP_\theta)|_{\theta' = \theta}(X))^2]$ and $\Delta > 0$. Then, under some regularity conditions (see Definition~\ref{defi:regularity}),
    \[D_\alpha(P_{\theta + \Delta},P_{\theta}) = \frac{\alpha}{2}I(\theta)\Delta^2 + o(\Delta^2).\]
\end{proposition}

The above proposition establishes an approximation of the Rényi divergence $D_\alpha(P_{\theta+\Delta},P_{\theta})$ in terms of the Fisher information $I(\theta)$. This result is particularly useful for families of distributions in which the Fisher information is substantially easier to compute than the Rényi divergence. However, this approximation is local, in the sense that it holds when the sensitivity $\Delta$ is small, and, to our knowledge, does not extend to global relationships between Rényi divergence and Fisher information \footnote{An existing attempt by \citet{Abbasnejad2006} contains a flaw in its proof; see Appendix~\ref{app:fisher-global-counterexample}, where we also clarify in Appendix~\ref{app:fisher-global-counterexample} why global comparison statements based on pointwise Rényi divergences and Fisher information cannot hold in general, and we give a basic counterexample.}.\looseness=-1

For privacy certification, one requires \emph{non-asymptotic} upper bounds. In this work, we proceed by introducing an upper envelope $U$ on the Rényi divergence $D_{2\alpha-1}(P_z,P_\theta)$ along a path $(\theta,\theta')$: $D_{2\alpha-1}(P_z,P_\theta)\leq U(z,\theta)$. Then, we show that for any such choice of $U$, a uniform bound on the Fisher information along the path yields an upper bound on Rényi divergence $D_\alpha(\theta',\theta)$. In particular, whenever $U$ is tractable, this provides a closed-form upper bound for $D_\alpha(\theta',\theta)$.

\begin{proposition}[Upper bounding Rényi divergences through Fisher information]
\label{prop:Fisher-upper-bound}
    Let $\theta < \theta' \in \Theta$ and $\alpha > 1$. We assume that there exists a function $U : \mathbb{R}^2 \to \mathbb{R^+}$ such that for all  $z \in (\theta,\theta')$, $D_{2\alpha-1}(P_{z},P_{\theta}) \leq U(z,\theta)$.  Then, under some regularity conditions (see Definition~\ref{defi:regularity}):
    \begin{align*}
        D_\alpha(P_{\theta'},P_{\theta}) \leq \frac{1}{\alpha-1}\log \bigg( 1 + \alpha \sup_{z \in (\theta,\theta')}I(z)^{1/2} \int_{\theta}^{\theta'} e^{(\alpha-1)U(z,\theta)}dz\bigg).
    \end{align*}
\end{proposition}

A discussion about the tightness of this upper bound can be found in Appendix~\ref{app:prop-Fisher-upper-bound}. We note that related considerations and proof techniques appear in~\citet{Karbowski2024} to bound the rates of several statistical divergences under temporal evolution of the distributions. In contrast, Proposition~\ref{prop:Fisher-upper-bound} yields a global bound after integrating along a parameter path, provided a computable upper envelope $U$ of $D_{2\alpha-1}$ is available.

In our setting,  this condition follows from the post-processing inequality. In fact, as synthetic data is a randomized function of the parameters, we have $D_{2\alpha-1}(P_v,P_w)\leq   D_{2\alpha-1}(V,W)$, yielding a closed-form upper envelope. This allows us to control the Rényi divergence in terms of the Fisher information. Crucially, if the Fisher information is small along a path, the upper bound forces $D_\alpha(P_{\theta'},P_\theta)$ to be small, which we will leverage to prove privacy amplification.

In summary, we have obtained two complementary tools to characterize privacy amplification: The local expansion in Proposition~\ref{prop:local-fisher} will yield sharp privacy amplification heuristics in regimes where the sensitivity $\Delta$ is small, while the global, non-asymptotic bound of Proposition~\ref{prop:Fisher-upper-bound} ensures rigorous guarantees. We will use it to show that privacy amplification holds for every value of $\Delta$.

\section{From Outputs to Sufficient Gram Statistics}
\label{sec:sufficient}

As explained in Section~\ref{sec:l_infty}, our goal is to quantify the privacy loss $D_\alpha(ZV,ZW)$ associated to synthetic data release, especially when $n_{\mathrm{syn}}$ is large. In this section, we show that the test between $ZV$ and $ZW$ can be reduced to a test between the Gram statistics $V^\top V$ and $W^\top W$ when $n_{\mathrm{syn}} \to +\infty$.
We first show the following result, enabled by the Gaussianity of $Z$.

\begin{restatable}{proposition}{propSufficientStats}
\label{prop:sufficient-statistics}
    Let $n_{\text{syn}}, d, k \geq 1$. Then, $D_\alpha(ZV,ZW) = D_\alpha(V^\top Z^\top Z V, W^\top Z^\top Z W)$.
\end{restatable}
The proof, provided in Appendix~\ref{app:prop-sufficient-statistics}, relies on showing that the map $T : X \mapsto X^\top X$ is a sufficient statistic for distinguishing between $ZV$ and $ZW$. The matrix $Z^\top Z$ follows a Wishart distribution $\mathcal{W}(d,n_{\mathrm{syn}})$ with shape $d$ and degree of freedom $n_{\mathrm{syn}}$. Then, a careful analysis allows us to connect the divergence between $ZV$ and $ZW$ to the divergence between $V^\top V$ and $W^\top W$ in the asymptotic regime.

\begin{restatable}{proposition}{propGramLimit}
\label{prop:gram-limit}
    Let $d, k \geq 1$. Then, $D_\alpha(V^\top Z^\top Z V, W^\top Z^\top Z W) \overset{n_{\mathrm{syn}}\to \infty}{\longrightarrow} D_\alpha(V^\top V, W^\top  W)$.
    Furthermore, for all $n_{\mathrm{syn}}> 1$, $D_\alpha(ZV,ZW) \leq D_\alpha(V^\top V, W^\top  W)$.
\end{restatable}

This key proposition allows us to obtain privacy guarantees for any $n_{\mathrm{syn}} \in \mathbb{N}^*$ by analyzing $D_\alpha(V^\top V, W^\top  W)$. The limit depends on $V,W$ only through the Gram matrices $V^\top V$ and $W^\top W$. This reduction is the core technical simplification that unlocks closed-form upper bounds for the divergence. In particular $D_\alpha(V^\top V, W^\top  W)$ is the privacy loss associated to releasing an infinite amount of synthetic data. In other words, from a privacy perspective, \textbf{the release of arbitrarily many synthetic data is equivalent to releasing the Gram matrix $V^\top V$}. We note that a sufficient-statistics perspective also appears in~\cite{Neunhoeffer2024}. However, they study a different setting: the generator is trained without formal privacy guarantees, with synthetic data obtained by sampling Gaussian noise that has the same empirical mean as the original dataset.

In the rest of the paper, we focus on deriving an upper bound for $D_\alpha(V^\top V, W^\top  W)$.
We first state a negative result: without additional assumptions, releasing infinitely many synthetic records provides no stronger privacy guarantees than releasing the generative model itself.

\begin{restatable}{proposition}{propNoFreeLunch}{\normalfont (No worst-case amplification in the $n_{\mathrm{syn}} \to + \infty$ regime).}
\label{prop:nofreelunch}
    \[\sup_{v,w \in \mathbb{R}^{d \times k}, \|v-w\|_F \leq \Delta}D_\alpha(V^\top V, W^\top  W) = D_\alpha(V,W).\]
\end{restatable}
The idea of the proof is to take $v_\theta = t e_1, w_t = (t + \Delta)e_1$, where $e_1 = (\delta_{11}) \in \mathbb{R}^{d\times k}$. As $t$ tends to $+\infty$, the matrices $V^\top V$ and $W^\top W$ diverge in a way that allows an adversary to recover $V$ and $W$.

\textbf{On pairwise and uniform RDP guarantees.} This proposition should not be interpreted as ruling out amplification in terms of uniform privacy guarantees.
Rather, it shows that any meaningful amplification theorem must depend on
the norm of the model parameters. Thus, our bounds will be stated as pairwise divergence bounds for the test between the synthetic data distributions induced by two adjacent datasets. They will be expressed depending on the norm of the parameters $\|v\|_F$ and $\|w\|_F$. We define $R(v,w) = \max\{\|v\|_F, \|w\|_F\}$. Our upper bounds are RDP guarantees whenever $R(v,w) \leq C$  for all adjacent $\mathcal{D}, \mathcal{D'}$.
Importantly, this condition is satisfied in the setting of~\cite{pierquin2025}, making our bounds applicable in their setting. In fact, their $l_2$-regularized linear regression setting naturally enforces parameter boundedness: $R(v,w) \leq O(1/\lambda)$, making their setting a special case of ours.
In more general settings, this is the standard way in which norm-sensitive bounds become accounting rules. Projection or parameter clipping enforces such a bound explicitly, and Projected Noisy Gradient Descent (PNGD) is commonly studied in theoretical DP works~\citep{Feldman2018,Altschuler2022}.
Moreover, some practical DP synthetic data generation algorithms also include an explicit weight clipping component at each step, such as DP-GANs~\cite{xie2018}, and weight decay is a classical component of ML pipelines.

\section{Privacy Amplification in Linear Synthetic Data Generation}
\label{sec:privacy-amp-linear}

In this section, we build on the results of Sections~\ref{sec:fisher-renyi} and~\ref{sec:sufficient} to derive our main privacy amplification guarantees for synthetic data generation with a linear generator.

Our objective is to evaluate $D_\alpha(V^\top V, W^\top W)$. In this multi-dimensional case, the matrix $V^\top V$ follows a non-central Wishart distribution $\mathcal{W}(k,d,v^\top v)$ with size $k$, degree of freedom $d$ and non-centrality parameter $\Omega_v = v^\top v$. As it does not have a standard designation, from now on, we refer to the parameter $v$ as the \emph{non-centrality parameter}, overriding $\Omega_v$. When $d \geq k$, $V^\top V$ admits a density~\citep{Muirhead2009}: $P_v(y) =  P_0(y) e^{-tr(\Omega_v)/2} {}_0F_1\Big(\frac{d}{2}, \frac{1}{4}y^{1/2} \Omega_v y^{1/2} \Big)$,
where $P_0$ denotes the density of the (central) Wishart distribution $\mathcal{W}(k,d)$ of size $k$ and with degree of freedom $d$, and ${}_0F_1$ is the generalized hypergeometric function of a matrix argument. The function ${}_0F_1$ can be characterized as a series involving zonal polynomials of a matrix argument. However, these latter polynomials are notoriously hard to handle~\citep{Koev2006}, making direct computation of $D_\alpha(V^\top V,W^\top W)$ intractable.

To our knowledge, there is no existing analysis in the literature characterizing the Fisher information of non-central Wishart distributions with respect to the non-centrality parameter. There exists an upper bound for the case $k=1$~\cite{Idier2014}, but it is too loose to be leveraged in our setting (more details about the existing approaches in Appendix~\ref{app:fisher-one-dim}). In this section, we tighten the upper bound of~\citet{Idier2014} in the $k=1$ case and provide the first Fisher information upper bounds for non-central Wishart distributions ($k > 1$).

We expect amplification to result from the contraction of the Rényi divergence through the Gram map $V \mapsto V^\top V$. In fact, non-central Wishart distributions $V^\top V$ are invariant by left orthogonal transformations $(Uv)^\top (Uv) = v^\top v$ of the parameter $v$, yielding a search space where $V$ and $W$ are distinguishable but not $V^\top V$ and $W^\top W$.

We denote the family of non-central Wishart distributions with respect to the non-centrality parameter as: \[P := \{P_\theta= \mathcal{W}(d,k, \theta^\top \theta); \; \theta \in \mathbb{R}^{d \times k}\},\] and the Fisher information associated with this family along a direction $H \in \mathbb{R}^{d \times k}$ as $$I(\theta)[H] :=  \mathbb{E}_{Y_\theta \sim p_\theta}\big[\big(\partial_t|_{t=0} (dP_{\theta+tH} /dP_\theta)(Y_\theta)\big)^2\big].$$

We also notice that $Y_\theta \sim P_\theta$ is obtained through a Markov kernel: \[X_\theta \sim Q_\theta :=\mathcal{N}(\theta,I_{d \times k}) \to X_\theta^\top X_\theta = Y_\theta \sim \mathcal{W}(d,k, \theta^\top \theta).\]

Then, the likelihood ratio writes: $\frac{dP_{\theta+tH}}{dP_\theta}(Y_\theta) = \mathbb{E}\left[ \exp\left( t \langle N, H \rangle_F - t^2\|H\|_F^2/2\right) \mid Y_\theta\right]$. This simplification is particularly useful because it allows to analyze the score without having to analyze ${}_0F_1$, which is intractable, nor assume that $d \geq k$.  Since conditional expectation is the projection onto $L^2(\sigma(Y_\theta))$ the space of $L^2$ variables that are measurable with respect to the sigma-algebra generated by $Y_\theta$, and leveraging $\mathbb{E}[\langle H, N \rangle_F] = 0$:
\[I(\theta)[H] = \sup_{F \in L^2(\sigma(Y_\theta)), \mathbb{E}[F] = 0, \mathbb{E}[F^2] > 0} \frac{\mathbb{E}[F \langle H, N \rangle_F]^2}{\mathbb{E}[F^2]}.\]

Therefore, in order to upper bound $I(\theta)[H]$, our proof strategy is the following. First, we let $F \in  L^2(\sigma(Y_\theta))$ such that $\mathbb{E}[F] = 0, \mathbb{E}[F^2] > 0$. We notice that $L^2(\sigma(Y_\theta)) \subset L^2(N)$ and therefore $F$ admits a Hermite expansion. We define $b_F = \mathbb{E}[NF]$ the first Hermite coefficient of $F$. Since $\mathbb{E}[F] = 0$, we have $\mathbb{E}[F \langle N, H \rangle_F] = \langle\mathbb{E}[NF], H \rangle_F = \langle b_F, H\rangle_F$. Therefore, our objective is to upper bound $\langle b_F, H\rangle_F / \mathbb{E}[F^2]$. The following theorem gives the desired upper bound:
\begin{restatable}{theorem}{thmUpperBoundFisher}\normalfont (Fisher information upper bound, non-central Wishart distributions).
\label{thm:Fisher-upper-bound}
    Let $d \geq 2$, $k \geq 1$ and $\theta, H \in \mathbb{R}^{d \times k}$. Then,
    \[I(\theta)[H] \leq \frac{2 \|\theta\|_F^2}{2 \|\theta\|_F^2 + d-1}\|H\|_F^2.\]
\end{restatable}

This upper bound nearly matches the lower bound of \citet{Idier2014} $2\|\theta\|_F^2/(2\|\theta\|_F^2 + d)$ specific to the $k=1$ case. Taken together, the two bounds tightly sandwich the Fisher information within the same functional form, up to a constant shift in the denominator. In particular, they capture the asymptotic behavior: $I(\theta)[H] \overset{d\to + \infty}{\longrightarrow} 0$ and an upper bound that tends to $\|H\|_F^2$ when $\|\theta\|_F \to +\infty$, yielding no amplification.
These bounds characterize the dependence of the Fisher information on both the degrees of freedom $d$ and the non-centrality parameter $\|\theta\|^2$. Strikingly, these bounds do not depend on $k$. They will play a crucial role in deriving upper bounds on the Rényi divergence between non-central Wishart distributions. Using the Fisher information bounds derived above and the local comparison principle in Proposition~\ref{prop:local-fisher}, we can derive our first key result: explicit upper bounds on the Rényi divergence $D_\alpha(V^\top V, W^\top W)$ by leveraging a parametrization $\widetilde P_t =  P_{w + t(v-w)/\|v-w\|_F}$ between $v$ and $w$.

\begin{restatable}{theorem}{thmUpperBoundRenyiLocal}\normalfont (Privacy amplification in the $n_{\mathrm{syn}} \to + \infty$ and
high privacy regime).
\label{thm:local-renyi}
    Let $d \geq 2$ and $v, w \in \mathbb{R}^{d \times k}$. Then the amplification factor $\eta = D_\alpha(V^\top V, W^\top W)/D_\alpha(V,W)$ admits the following upper bound:
    \[\eta \leq \frac{2R(v,w)^2}{2R(v,w)^2 + d-1} + o(1), \]
    where the $o(1)$ term denotes a quantity that tends to zero as the sensitivity $\Delta = \|v-w\|_F \to 0$,  while all other parameters remain fixed.
\end{restatable}
This result provides a dimension-dependent characterization of the divergence which is sharp in the radial regime. The rightmost inequality shows that the divergence is upper bounded by the factor $2R(v,w)^2/(2R(v,w)^2 + d-1)$ that quantifies the degree of privacy amplification. We observe two distinct behaviors. As $\|v\|,\|w\| \to +\infty$, the amplification factor upper bound tends to $1$, yielding no provable amplification, consistent with Proposition~\ref{prop:nofreelunch}. However, for fixed $v,w$, the amplification factor tends to $0$ at rate $O(d^{-1})$. 
While this bound is only relevant when the sensitivity is small, this is typically the case if parameters are obtained through regularized linear regression, where $\Delta$ is of order $1/n$~\citep{chaudhuri2011differentiallyprivateempiricalrisk}.

Nevertheless, we also derive a global upper bound based on our criterion in Proposition~\ref{prop:Fisher-upper-bound}.
\begin{restatable}{theorem}{thmUpperBoundRenyi}\normalfont (Privacy upper envelope in the $n_{\mathrm{syn}} \to + \infty$ regime).
\label{thm:upper-bound-renyi}
Let $d \geq 2$ and $v, w \in \mathbb{R}^{d \times k}$. Then,
    \begin{align*}
        D_\alpha(V^\top V,W^\top W) 
        \leq \frac{1}{\alpha-1}\log \left( 1+2\alpha \sqrt{\frac{2R(v,w)^2}{2R(v,w)^2 + d-1}}\frac{e^{(\alpha-1) (2\alpha-1) \frac{\Delta^2}{2}} -1}{(\alpha-1)(2\alpha-1)\Delta}\right).
    \end{align*}
\end{restatable}

Due to our proof technique, we obtain a $O(d^{-1/2})$ rate, relying on the asymptotic equivalence $\log(1+x) \sim x$ as $x \to 0$. However, Theorem~\ref{thm:local-renyi} suggests that a $O(d^{-1})$ exact rate may be achievable. This theorem yields amplification whenever the upper envelope is lower than the post-processing bound.

\section{Experiments}
\label{sec:experiments}
In this section, we perform two series of experiments to validate and illustrate our theoretical results.

\subsection{\texorpdfstring{Estimation of the privacy leakage $D_\alpha(ZV,ZW)$}{Estimation of the privacy leakage}}
\label{sec:experiments-theoretical}
We perform a series of experiments to illustrate our theoretical results. While sampling from $ZV$ or $V^\top V$ is straightforward, their densities are analytically intractable. Therefore, we rely on variational methods to estimate Rényi divergences. Specifically, we numerically estimate the privacy loss using the \emph{Convex-Conjugate Rényi Variational Formula}, introduced by~\citet[][(Theorem 2.1 therein]{birrell2023}. A more detailed discussion can be found in Appendix~\ref{app:experiments}. 

First, we estimate the divergence $D_\alpha(ZV,ZW)$ as a function of the number of released synthetic samples $n_{\text{syn}}$, in order to qualitatively assess its convergence behavior. Remarkably, Figure~\ref{fig:divergence-nsyn} shows that the Rényi divergence rapidly converges to a plateau, whose value is strictly lower than that of the corresponding post-processing privacy guarantee. Empirically, the number of synthetic samples required to reach the plateau appears to increase with $d$. The different values of the plateaus (dotted lines) highlight the amplification, which is significant even for the small values of $d \in (1,10)$ considered here. This empirical behavior illustrates our theoretical finding that privacy amplification persists under unlimited data release, and validates our focus on the infinite synthetic data release regime. We discuss the tightness of the infinite data release privacy loss $D_\alpha(V^\top V, W^\top W)$ with respect to $D_\alpha(ZV,ZW)$ from a theoretical point of view  in Appendix~\ref{subsec:convergence-rate}.

Then, we aim to estimate the empirical tightness of the local Rényi divergence bounds derived from Fisher information.
These bounds are asymptotic, and are thus expected to be informative only in the high privacy regime $\Delta \ll 1$. Specifically, we assess, for $\Delta \in (0,1)$, the extent to which the relation predicted by Theorem~\ref{thm:local-renyi} is satisfied. We estimate $D_\alpha(V^\top V, W^\top W)$  and compare the results to the theoretical envelope induced by Theorem~\ref{thm:local-renyi}, the lower part of the envelope being provided in Appendix~\ref{app:Local-fisher-renyi}.
The results, represented in Figure~\ref{fig:fisher-delta}, show that for small perturbations ($\Delta < 1$), the theoretical local bounds provide an accurate estimation of the amplification.

\begin{figure}[h]
\begin{minipage}[bc]{0.48\textwidth}
    \centering
    \includegraphics[width=\linewidth]{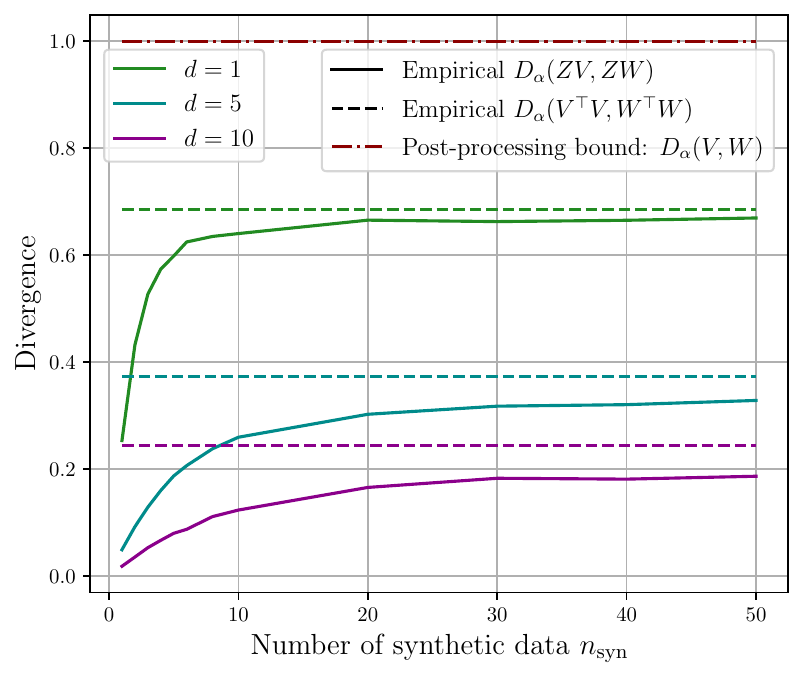}
    \caption{Empirical estimation of $D_\alpha(ZV,ZW)$ as a function of the number of released synthetic data $n_{\text{syn}}$ for multiple values of $d$, $k = \Delta = 1$, $C = \alpha = 2$.}
    \label{fig:divergence-nsyn}
\end{minipage}
\hfill
    \begin{minipage}[bc]{0.48\textwidth}
        \centering
    \includegraphics[width=\linewidth]{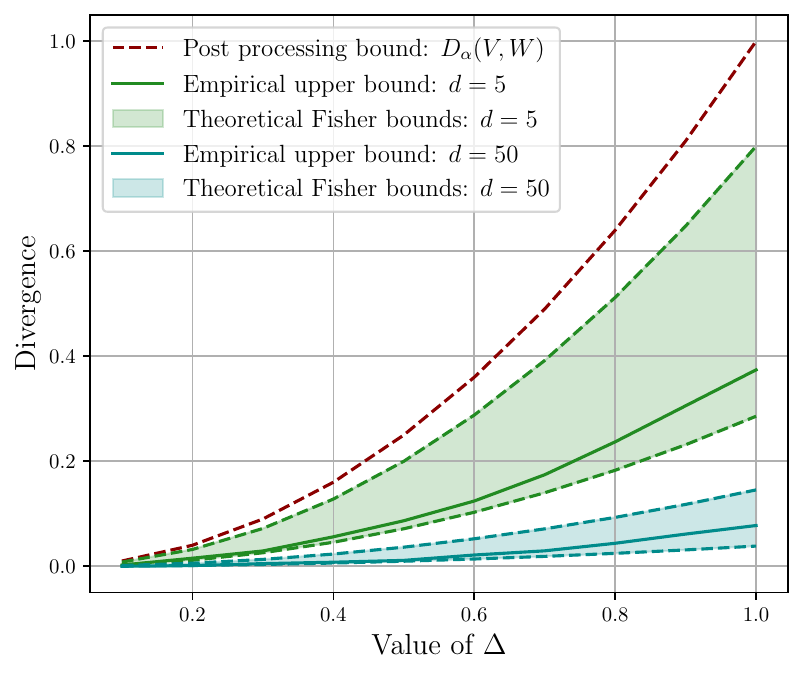}
    \caption{Empirical estimation of $D_\alpha(V^\top V,W^\top W)$ as a function of $\Delta$ for multiple values of $d$, $k = 1$, $\|w\|_F \geq 1$, $C = 2$.}
    \label{fig:fisher-delta}
    \end{minipage}
\end{figure}

\subsection{Privacy auditing of variational autoencoders trained with DP-SGD}
\label{sec:privacy-auditing}
We next test whether the qualitative amplification behavior predicted by our theory for linear generators also appears for nonlinear generators. We train two-layer variational autoencoders (VAEs) with DP-SGD on the worst-case datasets $\mathcal{D},\mathcal{D}'$ used by~\cite{Annamalai2024}. The VAE has latent dimension $d$, hidden dimension $h=64$ and output dimension $k=3$. All runs use the same DP-SGD training procedure and privacy budget. We do not impose parameter projection or explicit weight decay, so these experiments do not enforce the boundedness assumption used in our theory. 

 We perform the Querybased black-box attack described by~\cite{Annamalai2024}. The auditor is trained from $3000$ shadow models, with $1000$ validation models and $1000$ test models, to distinguish between synthetic datasets generated from trained on $\mathcal{D}$ from those trained on $\mathcal{D}'$. Since the data domain has only $2^3 = 8$ possible records, the attack can accurately estimate the distribution of the released synthetic datasets, which can be represented as a low-dimensional histogram, using the shadow models. This enables a strong and near-optimal audit of the information leakage induced by the synthetic data release. We report the attack AUC; lower AUC means weaker distinguishability, and AUC $= 0.5$ corresponds to a random guess.
 
 The audit shows two trends that are qualitatively consistent with our analysis. First, Figure~\ref{fig:privacy-auditing-nsyn} suggests that, as $n_{\mathrm{syn}}$ increases, the AUC quickly approaches a plateau, mirroring the convergence of $D_\alpha(ZV,ZW)$ to its infinite release in Figure~\ref{fig:divergence-nsyn}. Second, the plateau decreases as the latent dimension $d$ increases, matching the dimension-dependent amplification heuristic developed by our theory (Figure~\ref{fig:privacy-auditing-d}). We emphasize that these experiments do not provide formal privacy guarantees for VAEs, but suggest that the same qualitative amplification phenomenon may persist beyond linear generators. More details about the experiment can be found in Appendix~\ref{app:experiments}.

\begin{figure}[h]
\begin{minipage}[bc]{0.48\textwidth}
    \centering
    \includegraphics[width=\linewidth]{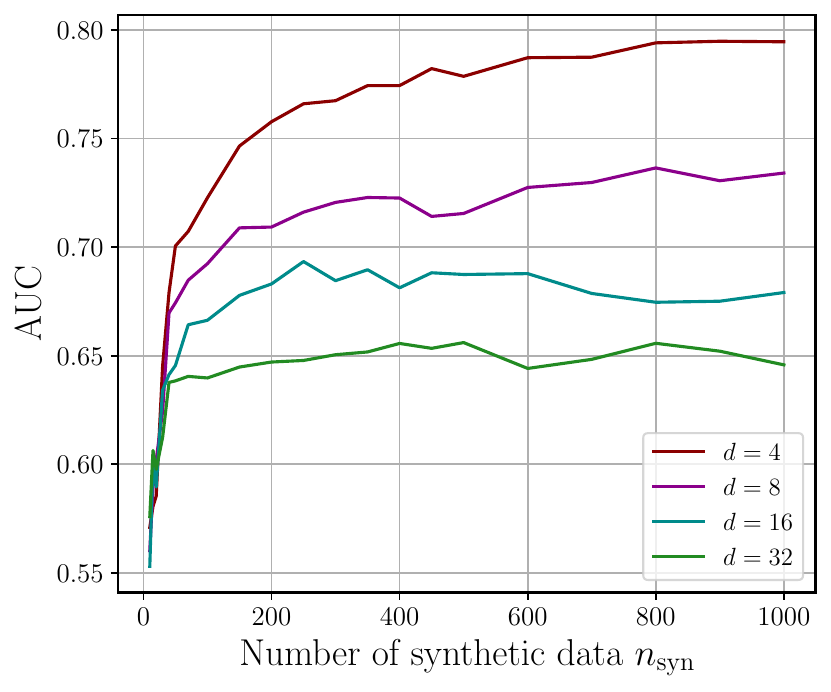}
    \caption{AUC Performance of the Querybased black-box attack as a function of $n_{\mathrm{syn}}$ for different values of $d$.}
    \label{fig:privacy-auditing-nsyn}
\end{minipage}
\hfill
    \begin{minipage}[bc]{0.48\textwidth}
        \centering
    \includegraphics[width=\linewidth]{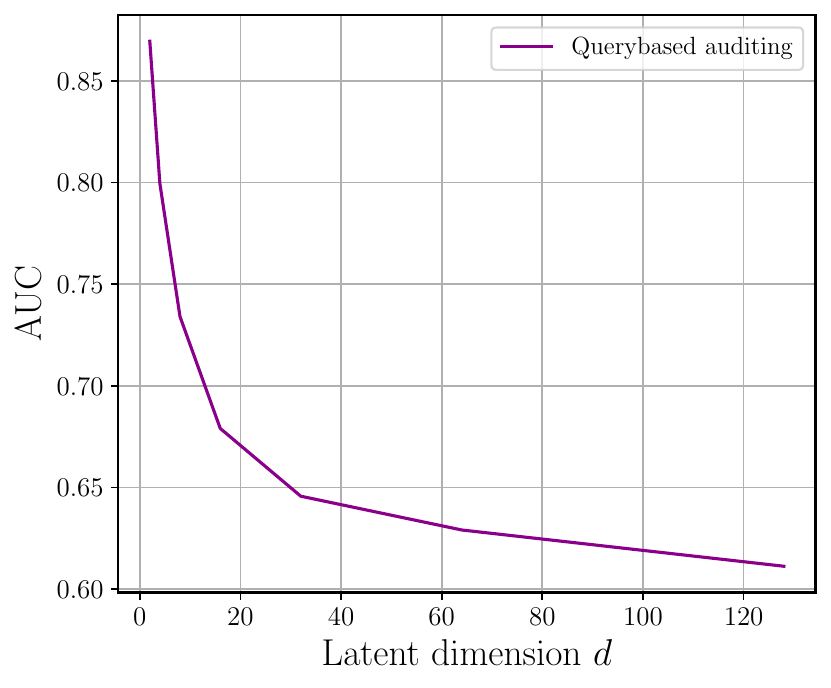}
    \caption{AUC Performance of the Querybased black-box attack as a function of $d$, $n_{\mathrm{syn}} = 1000$.}
    \label{fig:privacy-auditing-d}
    \end{minipage}
\end{figure}

\section{Conclusion and limitations}

We studied privacy amplification for synthetic data release by linear generators. Amplification can persist even when releasing an unbounded number of synthetic samples. Our analysis provides a unified characterization of the privacy loss that smoothly interpolates between classical post-processing guarantees and genuine amplification. Notably, it does not depend on the number of released synthetic records, improving on the work of~\citet{pierquin2025}. 

In the infinite-release regime, the privacy loss is fully characterized by Gram statistics $V^\top V$, which enables tractable bounds via Fisher information. This yields local Rényi DP upper bounds with tight amplification factors. Beyond the local regime, we derived a general upper bound based on a criterion relating Fisher information to Rényi divergences. While this bound is fully non-asymptotic and broadly applicable in our setting, we do not expect it to be tight in general. An important open question is whether this criterion can be improved to recover the local rates globally. 

More broadly, our results provide concrete analytical insights, such as the role of parameter boundedness, the Fisher information analysis, the explicit dependence of amplification on the model dimension, and the sufficient statistics reduction. We expect that these insights will provide both intuition and methodological guidance for understanding privacy amplification in more complex models and data release mechanisms.

The main limitation is that our formal guarantees apply to linear generators with output perturbation. Extending the framework to generators trained with practical private algorithms, and ultimately to nonlinear generative models such as DP-SGD-trained VAEs, is a key direction for future work; we discuss possible routes toward such extensions in Appendix~\ref{appsubsec:extensions}.

\section*{Acknowledgements}

This work was partially supported by Craft AI through a CIFRE PhD agreement. MT and AB are partially supported by the ANR-22-PECY-0002 IPoP (Interdisciplinary Project on Privacy) project of the Cybersecurity PEPR under the France 2030 program. CP and MB are employed by Craft AI. The authors declare no other competing interests.

\bibliography{references}
\bibliographystyle{apalike}

\newpage

\appendix

\onecolumn

\section*{LLM usage disclosure}

We used LLMs to assist with writing and grammar checking.
We also used LLMs as an auxiliary aid when reproducing the experiments of~\cite{Annamalai2024}, primarily for understanding implementation details and debugging reproduction code.

All theoretical claims, experimental results,
figures, references, and conclusions were independently verified by the authors.

\section*{Appendix}

This appendix provides detailed discussions, proofs and supplementary content.

\section{Background on differential privacy}
 
\label{app:DP}
This section briefly reviews the differential privacy notions and tools used in the paper. Throughout, $\mathcal{M}$ denotes a randomized mechanism. We call two datasets $\mathcal{D}$ and $\mathcal{D}'$ of the same size $n$ \emph{adjacent} if they differ in exactly one data record.

Rényi Differential Privacy quantifies privacy via the Rényi divergence between the output distributions induced by $\mathcal{M}$ on adjacent datasets~\citep{Mironov2017}.
\begin{definition}[Rényi Differential Privacy (RDP)~\citep{Mironov2017}]
     For $\alpha>1, \varepsilon > 0$, the Rényi divergence of order $\alpha$ between distributions $P$ and $Q$ is $D_\alpha(P,Q) := \frac{1}{\alpha-1}\log\mathbb{E}_{x \sim Q}\Big[\big(\frac{P(x)}{Q(x)}\big)^\alpha\Big]$. For $X \sim P, W \sim Q$, we note $D_\alpha(V,W) = D_\alpha(P,Q)$. Then, $\mathcal{M}$ satisfies $(\alpha,\varepsilon)$-RDP if for every pair of adjacent datasets $\mathcal{D},\mathcal{D}'$,
     \[D_\alpha(\mathcal{M}(\mathcal{D}) \| \mathcal{M}(\mathcal{D}')) \leq \varepsilon.\]
\end{definition}

The prior results of~\citet{pierquin2025} are formulated for trade-off functions in the $f$-Differential privacy ($f$-DP) framework. Trade-off functions~\citep{Dong2019} describe the optimal relationship between type I and type II errors for tests that attempt to distinguish the outputs produced from two adjacent datasets. Formally, let $P$ and $Q$ be two distributions and consider testing:
\[H_1: \text{the distribution is } P \text{~~ vs ~~} H_2: \text{the distribution is } Q.\]
For a  rejection rule $\phi\in[0,1]$, the type I error is $\mathbb{E}_P[\phi]$ and the type II error is $1-\mathbb{E}_Q[\phi]$.
\begin{definition}[Trade-off function]
    Let $t \in (0,1)$. Then, the trade-off function of the statistical test between $H_1$ and $H_2$ is defined as:
    \[T(P,Q)(t) = \inf_{\phi} \{1-\mathbb{E}_Q[\phi]:\mathbb{E}_P[\phi] \leq t\}.\]
    For $X \sim P, W \sim Q$, we note $T(V,W) = T(P,Q)$.
\end{definition}

$f$-Differential privacy is defined as an upper bound condition on the function $f$ for any trade-off function between the statistical test between adjacent datasets:
\[H_1: \text{the distribution is } \mathcal{M}(\mathcal{D}) \text{~~ vs ~~} H_2: \text{the distribution is } \mathcal{M}(\mathcal{D}').\]

\begin{definition}
    Let $f:[0,1]\to[0,1]$ be a decreasing convex function. A mechanism $\mathcal{M}$ satisfies \emph{$f$-differential privacy} ($f$-DP) if for all adjacent datasets $\mathcal{D},\mathcal{D}'$,
    \[T(\mathcal{M}(\mathcal{D}),\mathcal{M}(\mathcal{D}')) \geq f.\]
\end{definition}
\section{\texorpdfstring{Comparison of our upper bounds to the prior results of~\citet{pierquin2025}}{Comparison of our upper bounds to the prior results of Pierquin et al. (2025)}}

\label{app:comparison-prior-work}
In Section~\ref{sec:l_infty}, we mentioned that the privacy loss upper bounds obtained by~\citet{pierquin2025} follow the relationship:
\[D_\alpha(ZV,ZW) \lesssim \min\left\{1, k\sqrt{\frac{n_{\text{syn}}}{d-k}}\right\}D_\alpha(V,W).\]

We clarify that this upper bound is not a guarantee offered by~\citet{pierquin2025} (after conversion from trade-off functions to RDP). It represents the \emph{best scaling} one can hope to extract from their proof technique. This means that any conversion of the results of~\citet{pierquin2025} to RDP guarantees is at least as loose as this bound. In the rest of the section, we make precise this claim and highlight several limitations of these guarantees. 

The privacy loss upper bounds of~\citet{pierquin2025} are formulated in terms of trade-off functions. We restate below the result we use, with the notational change that the trade-off function argument is denoted by $t\in(0,1)$ (reserving $\alpha$ for Rényi orders).
\begin{theorem}[Privacy amplification for synthetic data, adapted from~\citep{pierquin2025}]
\label{theo:convergence-trade-off-multi-dim}

    Let $d > 0$. We denote $C_{n_{\text{syn}},k,d} =  C'k \sqrt{\frac{n_{\text{syn}}}{d-k}}$, where $C' > 0$ is an absolute constant. Let $N \in \mathbb{R}^{n_{\text{syn}} \times d}$ be a Gaussian matrix independent of $Z$, and define the trade-off function: \[G_{(d-k,v,w)} = T\left(\sqrt{d-k}N + Zv,  \sqrt{d-k}N + Zw\right),\]
    with $G_{(d-k,v,w)}(t) = 0$ for $t > 1$.
    
    Then, for all $t \in (0,1)$:
    \begin{equation}
    \label{equa:prior-tf-bound}
        f:=\max\left\{T(V,W),G_{(d-k,v,w)}\left(t + C_{n_{\text{syn}},k,d}\right) -C_{n_{\text{syn}},k,d}\right\} \leq T(ZV,ZW)(t).
    \end{equation}
\end{theorem}

Specifically, we show that:
\begin{itemize}
    \item Proposition~\ref{prop:prior-dp-delta}: Any $(\varepsilon,\delta)$-DP conversion of the amplification term $G_{(d-k,v,w)}\left(t + C_{n_{\text{syn}},k,d}\right) -C_{n_{\text{syn}},k,d}$ of Equation~\ref{equa:prior-tf-bound} yields $\delta = \Omega\left(k(n_{\text{syn}}/d)^{1/2}\right)$. In particular, to get $\delta = o(n^{-1})$ (with $n$ the sample size), one would need $n_\text{syn} = o\left(\frac{d}{k^2 n^2}\right)$.
    \item Proposition~\ref{prop:prior-no-amplification}: Equation~\ref{equa:prior-tf-bound} does not yield any amplification in the high privacy regime $\Delta  < 2C'k \sqrt{\frac{2\pi n_{\text{syn}}}{d-k}}$. In this regime, the privacy loss is no better than $D_\alpha(V,W)$. 
    \item Proposition~\ref{prop:prior-rate}: In the privacy amplification regime, the Rényi divergence conversion $l_\alpha(f)$ of the trade-off function $f$ derived from Equation~\ref{equa:prior-tf-bound} scales as $2C'k \sqrt{\frac{ n_{\text{syn}}}{d-k}}$: \[l_\alpha(f) \geq 2C'k \sqrt{\frac{ n_{\text{syn}}}{d-k}} D_\alpha(V,W).\]
\end{itemize}

First, we explain our proof strategy to derive RDP bounds from Equation~\ref{equa:prior-tf-bound}. Figure~\ref{fig:lower-tradeoff-function} represents the associated trade-off function. This function is piecewise, separated into two regions: a Gaussian trade-off region and an amplification region. By symmetry, the trade-off function is Gaussian at the edges, in the set $(0,x_-)\cup (x_+,1)$ for some $0 \leq x_- \leq x_+ \leq 1$, and $G_{(d-k,v,w)}\left(t +  C_{n_{\text{syn}},k,d}\right) -C_{n_{\text{syn}},k,d}$ at the center, in the interval $(x_-,x_+)$. Then, Propositions~\ref{prop:prior-dp-delta} and~\ref{prop:prior-no-amplification} become clear: any $(\varepsilon,\delta)$ conversion of the amplification regime yields $\delta \geq C_{n_{\text{syn}},k,d}$. Moreover, there is no amplification if the curve $G_{(d-k,v,w)}\left(t + C_{n_{\text{syn}},k,d}\right) -C_{n_{\text{syn}},k,d}$ does not intersect the curve $T(V,W)$. We identify the regimes where it happens.

\begin{figure}[h]
    \centering
    \includegraphics[width=0.6\linewidth]{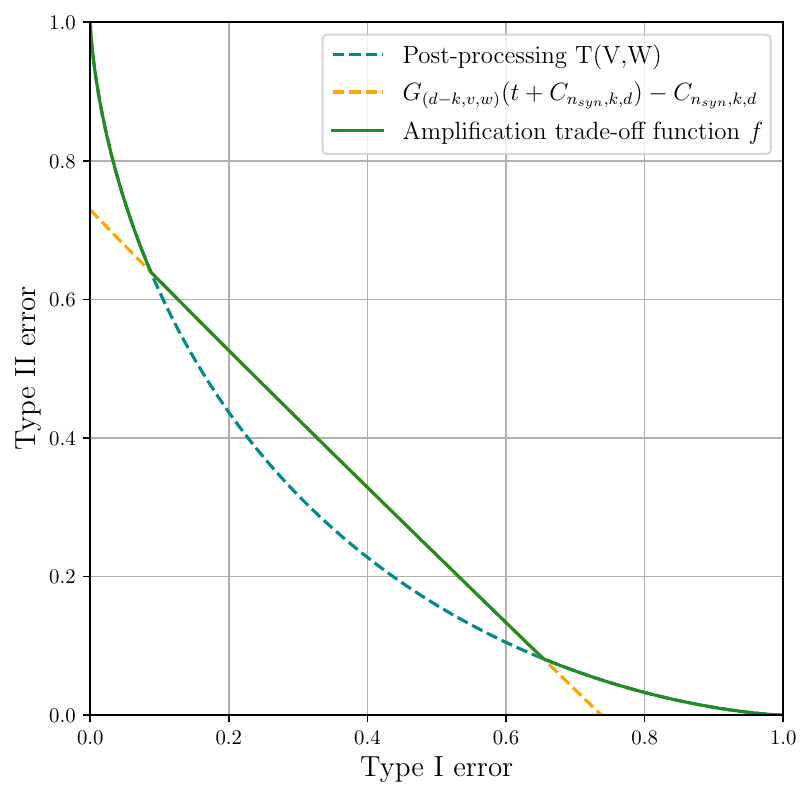}
    \caption{Representation of the trade-off function $f$ for $\Delta=1, C' = 1, d = 60, n_{\text{syn}} = 1, k=1$.}
    \label{fig:lower-tradeoff-function}
\end{figure}

\textbf{Relaxing the trade-off function $f$}:
We rely on a simplification of the $G_{(d-k,v,w)}\left(t + C_{n_{\text{syn}},k,d}\right) -C_{n_{\text{syn}},k,d}$ part. Let $t \in (0,1)$. From Equation~\ref{equa:prior-tf-bound}, $T(ZV,ZW)(t) \geq f(t)$. $f$ is the certifiable privacy loss upper bound obtained by~\citet{pierquin2025}. Then, for any trade-off function $h > f$, the analysis of~\citet{pierquin2025} is not sufficient to certify privacy. In particular, any trade-off function $T$ satisfies $T(t) \leq 1-t$. As $G_{(d-k,v,w)}$ is a trade-off function, we have $G_{(d-k,v,w)} \leq 1- I_d$. Then, \[f \leq \max\left\{T(V,W),1-I_d - 2C_{n_{\text{syn}},k,d}\right\} = \max\left\{\Phi(\Phi^{-1}(1-\alpha) - \Delta),1-I_d - 2C_{n_{\text{syn}},k,d}\right\} := g,\]  where $\Phi$ is the Gaussian cumulative distribution function.

Replacing $G_{(d-k,v,w)}\left(t + C_{n_{\text{syn}},k,d}\right) -C_{n_{\text{syn}},k,d}$ by $1-I_d - 2C_{n_{\text{syn}},k,d}$ unlocks simpler Rényi divergence conversion. This yields lower bounds and minimal conditions for amplification. 

\textbf{A useful reparametrization}: Now, assume that the curves $\Phi(\Phi^{-1}(1-\cdot ) -\Delta)$ and $1-I_d - 2C_{n_{\text{syn}},k,d}$ intersect (which means that there is amplification). Note $z  = \Phi^{-1}(1-t)$. Then, $f(z) = \max\{\Phi(z - \Delta), \Phi(z) - 2C_{n_{\text{syn}},k,d}\}$. The intersection is located at the solutions of $\Phi(z) - \Phi(z-\Delta) = 2C_{n_{\text{syn}},k,d}$. We note $z_- \leq 0 \leq z_+$ the lower and upper solutions. 

\textbf{Converting trade-off function into RDP guarantees}:
From~\citet{Dong2019}, trade-off functions convert into Rényi divergences through the following functional:

\begin{proposition}[Conversion from $f$-DP to RDP~\cite{Dong2019}]
If a mechanism is $f$-DP, then it is $(\alpha,l_\alpha(f))$-RDP for all $\alpha >1$ with:
\[l_\alpha(f) = \begin{cases}
    \frac{1}{\alpha-1}\log \int |f'(t)|^{1-\alpha} dt \text{ if } z_f=1,\\ +\infty \text{ else,}
\end{cases}\]
with $z_f = \inf\{t \in (0,1); f(t) = 0\} =1$.
\end{proposition}

We can now prove Propositions~\ref{prop:prior-dp-delta},~\ref{prop:prior-no-amplification} and~\ref{prop:prior-rate}:

\begin{proposition}[DP conversion forces $\delta$ of order $C_{n_{\text{syn}},k,d}$]
\label{prop:prior-dp-delta}
Any $(\varepsilon,\delta)$-DP guarantee derived from the amplification part of Equation~\ref{equa:prior-tf-bound}
necessarily satisfies $\delta=\Omega(C_{n_{\text{syn}},k,d})$ (uniformly over $\varepsilon > 0$). Therefore, $n_\text{syn} = o\left(\frac{d}{k^2 n^2}\right)$.
\end{proposition}

\begin{proof}
    Immediate from the form of the trade-off function and the $(\varepsilon,\delta)$ conversion result of~\citet{Dong2019}. Therefore, for a reasonable privacy protection $\delta = o(n^{-1})$, we have $n^{-1} > C_{n_{\text{syn}},k,d}$, giving the desired result. Moreover, converting the trade-off function of Equation~\ref{equa:prior-tf-bound} either yields no amplification (in the Gaussian region) or large $\delta$ (in the amplification region).
\end{proof}

\begin{proposition}[No amplification in the high-privacy regimes]
\label{prop:prior-no-amplification}
If $\Delta < 2C'k \sqrt{\frac{2\pi n_{\text{syn}}}{d-k}}$, the RDP conversion $l_\alpha(f)$ of Equation~\ref{equa:prior-tf-bound} yields $l_\alpha(f) = D_\alpha(V,W)$. In particular, $G_{(d-k,v,w)}\left(t + C_{n_{\text{syn}},k,d}\right) -C_{n_{\text{syn}},k,d}$ does not intersect $T(V,W)$.
\end{proposition}

\begin{proof}
The function $H = \Phi(z) - \Phi(z-\Delta)$ is symmetric around $\Delta/2$. In fact,
\begin{align*}
    H\left(\frac{\Delta}{2} + z\right) &= P\left(N \leq \frac{\Delta}{2} + z\right) - P\left(N \leq z - \frac{\Delta}{2} \right)\\
     &= P\left(N \geq -\frac{\Delta}{2} - z\right) - P\left(N \geq - z + \frac{\Delta}{2} \right)\\
     &=  P\left(N \leq \frac{\Delta}{2} - z \right) - P\left(N \leq -\frac{\Delta}{2} - z\right)= H\left(\frac{\Delta}{2} - z\right).
\end{align*}
It is also increasing then decreasing. In fact, $H'(z) = \frac{1}{\sqrt{2\pi}}(e^{-z^2/2} - e^{-(z-\Delta)^2/2}) \geq 0 \iff z \leq \frac{\Delta}{2}$.
Its maximum value is attained for $z = \frac{\Delta}{2}$. Therefore, the intersection between $\Phi(\Phi^{-1}(1-t) - \Delta)$ and $ (1-t - 2C_{n_{\text{syn}},k,d})$ exists if and only if $H(\Delta/2) \geq 2C_{n_{\text{syn}},k,d}$.

Also, $H(\Delta/2) = \Phi(\Delta/2) - \Phi(-\Delta/2) = 2\Phi(\Delta/2)-1$.

Then, there is no amplification for:
\[2C_{n_{\text{syn}},k,d} > \frac{\Delta}{\sqrt{2\pi}} \geq \sqrt{\frac{2}{\pi}}\int_0^{\Delta/2} e^{-t^2/2} dt = 2\Phi(\Delta/2) -1.\]
\end{proof}

Now, we show that the conversion of the trade-off function $f$ into Rényi divergence $l_\alpha(f)$ satisfies $l_\alpha(f) \geq C_{n_{\text{syn}},k,d} D_\alpha(V,W)$ when $C_{n_{\text{syn}},k,d}$ is sufficiently small.

\begin{proposition}
\label{prop:prior-rate}
    Assume that $d > n_{\text{syn}}, k$. Then, in the privacy amplification regime, or equivalently, when $G_{(d-k,v,w)}\left(t + C_{n_{\text{syn}},k,d}\right) -C_{n_{\text{syn}},k,d}$ intersects with $T(V,W)$, the following relationship holds:
    \[l_\alpha(f) \geq 2C'k\sqrt{\frac{n_{\text{syn}}}{d-k}}D_\alpha(V,W).\]
\end{proposition}

\begin{proof}

We give a lower bound of the Rényi divergence associated to the trade-off function $f$. We decompose the integral into $3$ pieces representing the Gaussian regime and the amplification regime, and perform the change of variable $z = \Phi^{-1}(1-t)$.
\begin{align*}
    \exp((\alpha-1)l_\alpha(f)) &=  \int_0^1 |f'(t)|^{1-\alpha} dt\\ &=  \int_0^1 |\partial_\theta\max\left\{\Phi(\Phi^{-1}(1-t) - \Delta),1-t - C_{n_{\text{syn}},k,d}\right\}|^{1-\alpha} dt\\
    &=  \int_0^{1-\Phi(z_+)} |\partial_\theta \Phi(\Phi^{-1}(1-t) - \Delta)|^{1-\alpha}dt \\&+ \int_{1-\Phi(z_-)}^1  |\partial_\theta \Phi(\Phi^{-1}(1-t) - \Delta)|^{1-\alpha}dt + |\Phi(z_-) - \Phi(z_+)|\\
    &=  \int_0^{1-\Phi(z_+)}-(\Phi^{-1})'(1-t)^{1-\alpha} \phi(\Phi^{-1}(1-t) - \Delta)^{1-\alpha}dt \\&+ \int_{1-\Phi(z_-)}^1 -(\Phi^{-1})'(1-t)^{1-\alpha}\phi(\Phi^{-1}(1-t) - \Delta)^{1-\alpha}dt +|\Phi(z_-) - \Phi(z_+)|\\
    &= \int_{z_+}^{+\infty} \frac{\phi(z)^{\alpha}}{\phi(z-\Delta)^{\alpha-1}}dz \\
    &+ \int_{-\infty}^{z_-} \frac{\phi(z)^{\alpha}}{\phi(z-\Delta)^{\alpha-1}}dz +|\Phi(z_-) - \Phi(z_+)|\;,\; (z = \Phi^{-1}(1-t)),
\end{align*}
where $\phi$ is the Gaussian density function. By completing the square, for $x < y$, we have:
\begin{align*}
    \int_{x}^{y} \frac{\phi(z)^{\alpha}}{\phi(z-\Delta)^{\alpha-1}}dz &= \frac{1}{\sqrt{2\pi}}\int_{x}^{y} e^{((\alpha-1)(z-\Delta)^2-\alpha z^2)/2}dz\\
    &= \frac{1}{\sqrt{2\pi}}\int_{x}^{y} e^{(-z^2 -2(\alpha-1)\Delta + (\alpha-1)\Delta^2)/2}dz\\
    &= \frac{1}{\sqrt{2\pi}}\int_{x}^{y} e^{-(z + (\alpha-1)\Delta)^2)/2}dz\\
    &=\frac{1}{\sqrt{2\pi}}e^{\alpha(\alpha-1)\Delta^2/2}\int_{x}^{y} e^{(-z^2 -2(\alpha-1)\Delta + (\alpha-1)\Delta^2)/2}dz\\
    &= e^{\alpha(\alpha-1)\Delta^2/2} (\Phi(y + (\alpha-1)\Delta) - \Phi(x + (\alpha-1)\Delta)).
\end{align*}

Then, we can express $\exp((\alpha-1)l_\alpha(f))$ as a function of the Gaussian cumulative function distribution $\Phi$:
\begin{align*}
    \exp((\alpha-1)l_\alpha(f)) &= e^{\alpha(\alpha-1)\Delta^2/2} (1 - \Phi(z_+ + (\alpha-1)\Delta) + \Phi(z_- + (\alpha-1)\Delta)+ |\Phi(z_-) - \Phi(z_+)|\\
    &= e^{\alpha(\alpha-1)\Delta^2/2} (\Phi(z_- + (\alpha-1)\Delta) +\Phi(-(z_+ + (\alpha-1)\Delta)) +  \Phi(z_+) - \Phi(z_-)
\end{align*}

Now, we derive an upper bound to $\Phi(z_+) - \Phi(z_-)$. We note $G_\Delta : x \mapsto \Phi(\Phi^{-1}(1-x) - \Delta)$, and $x_+ = 1-\Phi(z_+), x_- = 1- \Phi(z_-)$. By symmetry of $G_\Delta$ ($G_\Delta = G_\Delta^{-1}$), $x_+ = G_\Delta(x_-)$ and, since $x_-$ is at the intersection between $G_\Delta$ and $1 - I_d - 2C_{n_{\text{syn}},k,d}$:
\[x_+ + x_- = G_\Delta(x_-) + x_- = 1 - 2C_{n_{\text{syn}},k,d}.\]

Also, leveraging the definitions of $x_-$ and $x_+$:
\[\Phi(z_+) + \Phi(z_-) = (1-x_+) + (1-x_-) = 1+2C_{n_{\text{syn}},k,d}.\]

Then, we show that $z_+ = z_- + \Delta$. In fact, $\Phi(z_+) - \Phi(z_+-\Delta) = 2C_{n_{\text{syn}},k,d}$ and: 
\begin{align*}
    \Phi(\Delta - z_+) - \Phi(-z_+) = \Phi(z_+) - \Phi(z_+ - \Delta) = 2C_{n_{\text{syn}},k,d}.
\end{align*}
Then, if $z$ is a solution of the equation $\Phi(z) - \Phi(z-\Delta) = 2C_{n_{\text{syn}},k,d}$ , $\Delta-z$ is also a solution of the equation. Since this equation only admits two solutions, $z_+ + z_- = \Delta$.
We note $A = \Phi(z_- + (\alpha-1)\Delta) +\Phi(-(z_+ + (\alpha-1)\Delta)$ and $L = \Phi(z_+) - \Phi(z_-)$. Then, we have shown that:
\begin{align*}
        \exp((\alpha-1)l_\alpha(f)) &= e^{\alpha(\alpha-1)\Delta^2/2} A+ L.
\end{align*}
Now, let us get a lower bound of this expression. Since $\Phi$ is increasing and $\alpha > 1$, we have:
\begin{align*}
    A = \Phi(z_- + (\alpha-1)\Delta) +\Phi(-(z_+ + (\alpha-1)\Delta) \geq \Phi(z_- + (\alpha-1)\Delta) \geq \Phi(z_-)  \geq 2C_{n_{\text{syn}},k,d}.
\end{align*}
Also, we obtain a lower bound of $A+L$ by showing that $(\Phi(z_- + (\alpha-1)\Delta)- \Phi(z_-)) -(\Phi(z_+ + (\alpha-1)\Delta) - \Phi(z_+)) \geq 0$.
\begin{align*}
    A+L &= 1+ (\Phi(z_- + (\alpha-1)\Delta)- \Phi(z_-)) -(\Phi(z_+ + (\alpha-1)\Delta) - \Phi(z_+))\\
    &= 1+ \int_{0}^{ (\alpha-1)\Delta} (\phi(t+z_-) -  \phi(t+z_+)) dt\\
    &= 1+ \frac{1}{\sqrt{2\pi}}\int_{0}^{ (\alpha-1)\Delta} e^{((t+z_+)^2 - (t+z_-)^2)/2} dt\\
    &= 1+ \frac{1}{\sqrt{2\pi}}\int_{0}^{ (\alpha-1)\Delta} e^{(z_+-z_-)(z_+ + z_- + 2t)/2} dt \geq 1.\\
\end{align*}
By the same reasoning,
\[0 \leq A = 1+ \Phi(z_- + (\alpha-1)\Delta) -\Phi(z_+ + (\alpha-1)\Delta) = 1 + \int_{z_+}^{z_-} \phi(t + (\alpha-1)\Delta) dt \leq 1.\]

Then, by convexity of the $x \mapsto e^x$ function for the input $A \alpha(\alpha-1)\Delta^2/2 + (1-A)\cdot 0$:
\[e^{\alpha(\alpha-1)\Delta^2/2} A+ L \geq e^{\alpha(\alpha-1)\Delta^2/2} A+ 1-A \geq e^{A\alpha(\alpha-1)\Delta^2/2} \geq  e^{C_{n_{\text{syn}},k,d}\alpha(\alpha-1)\Delta^2}.\]
Recall that $D_\alpha(V,W) = \alpha\Delta^2/2$. Then, 
\[\exp((\alpha-1)l_\alpha(f)) \geq \exp(2(\alpha-1) C_{n_{\text{syn}},k,d} D_\alpha(V,W)).\]
Then, the desired result holds:
\[l_\alpha(f) \geq 2C'k\sqrt{\frac{n_{\text{syn}}}{d-k}}D_\alpha(V,W).\]
\end{proof}

\newpage
\section{Relationship between Rényi divergences and Fisher information}
\label{app:Renyi-Fisher}

In this section, we introduce Fisher information and its properties needed for the understanding of the paper. Then, we examine the relationship between Rényi divergences and Fisher information.

\subsection{Fisher information and its properties}
\label{app:Fisher-properties}
This appendix briefly recalls the notion of Fisher information and the few properties we rely on throughout the paper. We define Fisher information and score functions:

\begin{definition}[Fisher information and score function]
    Let $P = \{P_\theta : \theta\in\Theta \subset \mathbb{R}\}$ family of probability measures. Then, when the $P_\theta$ admit densities $p_\theta$ with respect to the Lebesgue measure, the score function is defined by:
    \[s_\theta(x) = \partial_\theta \log p_\theta(x).\]
    When for all $\theta, \theta' \in \Theta$, $P_\theta$ and $P_{\theta'}$ are mutually absolutely continuous, we can equivalently define the score as:
    \[s_\theta(x) = \partial_{\theta'} \left(\frac{d P_{\theta'}}{d P_\theta}\right)\mid_{\theta' = \theta} (x).\]
    The Fisher information of the family $P$ with respect to the parameter $\theta \in \Theta$ is:
    \[I(\theta) = \mathbb{E}_{X \sim P_\theta}[s_\theta(X)^2] \in \mathbb{R}^+.\]
\end{definition}

In this paper, we focus on Fisher information of families of probabilities indexed by scalar parameters. In this paper, all considered distributions are "sufficiently regular". We define sufficient regularity conditions for the usual Fisher information properties to hold.

\begin{definition}[Regular family of probability measures]
\label{defi:regularity}
     Let $P = \{P_\theta : \theta\in\Theta \subset \mathbb{R}\}$ be a family of probability measures on a measurable space $(E,\mathcal{F})$. Then, $P$ is regular if:
     \begin{itemize}
         \item $\Theta$ is a right-open set of $\mathbb{R}$: for every point $\theta \in \Theta$, there exists a neighborhood of the form $[\theta,\theta + \delta) \subset \Theta$.
         \item For every $\theta, \theta' \in \Theta$, $P_\theta$ is absolutely continuous with respect to $P_{\theta'}$ and conversely.
         \item For every $q \geq 1$, $\theta, \theta' \in \Theta$ and $r_{\theta',\theta} := \frac{dP_{\theta'}}{dP_\theta}$, $\theta' \mapsto r_{\theta',\theta}$ is $C^2$ as a function from $\Theta$ to $L^q(P_\theta)$: there exists $L^q(P_\theta)$ valued functions $\theta' \mapsto \dot r_{\theta',\theta}$ and $\theta' \mapsto \ddot r_{\theta',\theta}$ such that:
         \begin{align*}
             \left\| \frac{r_{\theta' + \delta, \theta} - r_{\theta',\theta}}{\delta} - \dot r_{\theta',\theta}\right\|_{L^q(P_\theta)} &\overset{\delta \to 0}{\longrightarrow} 0,\\
             \left\| \frac{\dot r_{\theta' + \delta, \theta} - \dot r_{\theta',\theta}}{\delta} - \ddot r_{\theta',\theta}\right\|_{L^q(P_\theta)} &\overset{\delta \to 0}{\longrightarrow} 0.
         \end{align*}
     \end{itemize}
\end{definition}
Note that these are not minimal conditions for the existence of Fisher information.

Here are some properties of the Fisher information that are used throughout the paper:

\begin{proposition}[Fisher information and score properties]
\label{prop:Fisher-properties}
    Let $P = \{P_\theta : \theta\in\Theta \subset \mathbb{R}\}$ be regular of probability  measures such that the Fisher information $I(\theta)$ exists for all $\theta \in \Theta$. Then, the Fisher information and the score function satisfy the following properties:
    \begin{itemize}
        \item \textbf{Mean-zero score and information identity, Lemma 5.3 of~\citet{Lehmann1998}:} $\mathbb{E}_{X \sim P_\theta}[s_\theta(X)] = 0$, hence $I(\theta) = \Var_{X \sim P_\theta}(s_\theta(X))$. Moreover, when $\partial_\theta^2 \log p_\theta(X)$ exists almost everywhere, the information identity holds: \[I(\theta) = - \mathbb{E}_{X \sim P_\theta}[\partial_\theta^2 \log p_\theta(X)].\]
        \item \textbf{Post-processing inequality~\citep{Schervish1995}:} Let $X\sim P_\theta$ and let $Y \sim K(X)$ be generated from a Markov kernel $K$ applied on $X$. Then, \[I_Y(\theta) \leq I_X(\theta).\]
    \end{itemize}
\end{proposition}

\subsection{A local relationship between Rényi divergences and Fisher information}
\label{app:Local-fisher-renyi}
In parametric distributions, the second order derivative of most $f$-divergences coincide the Fisher information~\citep{Polyanskiy2025}. Rényi divergences are obtained as a smooth monotone transformation of a $f$-divergence. This yields a local expansion of the second order of Rényi divergences as a function of Fisher information under small perturbations of the parameter for sufficiently regular families of distributions, as noted by~\citet{Haussler1997} and~\citet{vanErven2014}. However, in these works, the exact technical conditions on the parametrization that are needed are not explicited, and the authors do not know references for these conditions~\citep{vanErven2014}. For completeness, we prove that under the regularity conditions of Definition~\ref{defi:regularity}, the local expansion of Rényi divergences as a function of Fisher information hold. In this paper, all considered families of distributions satisfy these conditions. Note that we do not claim that these conditions are minimal.

\begin{proposition}[Local relationship between Rényi divergences and Fisher information]
\label{prop:local-fisher-app}
    Let $\{P_\theta, \theta\in \Theta \subset \mathbb{R}\}$ be a family of probability distributions. 
    Let $\theta \in \Theta$. Let $\Delta > 0$ such that $\theta + \Delta \in \Theta$. Then, if $P$ is regular (Definition~\ref{defi:regularity}),
    \[D_\alpha(P_{\theta + \Delta},P_{\theta}) = \frac{\alpha}{2}I(\theta)\Delta^2 + o(\Delta^2).\]
\end{proposition}

\begin{proof}
    Let $\alpha > 1$ and $P$ be a regular family of distribution, as defined in Definition~\ref{defi:regularity}. For $\theta,\theta' \in \Theta$, we note $H(\theta',\theta) = \mathbb{E}_{Y_\theta \sim P_\theta}\left[\left(\frac{dP_{\theta'}}{dP_{\theta}}\right)^{\alpha}(Y_\theta)\right] = \mathbb{E}_{Y_\theta \sim P_\theta}\left[r_{\theta', \theta}^\alpha\right]$. Fix $\theta \in \Theta$. Let $\Delta > 0$ such that $\theta + \Delta \in \Theta$. We note $s_{\theta'} = \partial_{\theta'} r_{\theta',\theta}|_{\theta' = \theta}$ and $a_{\theta'} = \partial_{\theta'}^2  r_{\theta',\theta}|_{\theta' = \theta}$. Since $P$ is regular, the conditions of Definition~\ref{defi:regularity} ensure twice differentiability of $r_{\theta',\theta}$ under the expectation.
    
    Then, by definition, $\mathbb{E}_{P_\theta}[r_{\theta',\theta}] = 1$, and differentiating twice: $\mathbb{E}_{P_\theta}[s_\theta] = E_{P_\theta}[a_\theta] = 0.$
    Also, we write the Taylor expansion of $r_{\theta +\Delta,\theta}$:

    \[r_{\theta + \Delta} = 1 + \Delta s_\theta + \frac{\Delta^2}{2}a_\theta + o_{L^q(P_\theta)}(\Delta^2).\]

    Then, we have:

    \[r_{\theta + \Delta}^\alpha = 1 + \alpha \Delta s_\theta + \frac{\alpha \Delta^2}{2}(a_\theta + (\alpha-1)s_\theta^2) + o_{L^q(P_\theta)}(\Delta^2).\]

    Taking the expectation under $P_\theta$:
    \[\mathbb{E}_{P_\theta}[r_{\theta + \Delta}^\alpha] = 1 + \alpha \Delta \mathbb{E}_{P_\theta}[s_\theta] + \frac{\alpha \Delta^2}{2}(\mathbb{E}_{P_\theta}[a_\theta] + (\alpha-1)\mathbb{E}_{P_\theta}[s_\theta^2]) + o(\Delta^2).\]

    Then, since $r_{\theta + \Delta,\theta} \overset{\Delta \to 0}{\to} 1$ in $L^q(P_\theta)$ for every finite $q$ and the likelihood ratios are locally bounded in $L^q(P_\theta)$, we obtain from Taylor's formula for $x^\alpha$: 

    \[H(\theta+\Delta, \theta) = \mathbb{E}_{P_\theta}[r_{\theta + \Delta}^\alpha] = 1 + \frac{\alpha (\alpha-1)\Delta^2}{2} I(\theta) + o(\Delta^2).\]

    Hence, $D_\alpha(P_{\theta + \Delta}, P_\theta) = \frac{1}{\alpha-1} \log H(\theta + \Delta, \theta) = \frac{\alpha}{2} I(\theta) \Delta^2 + o(\Delta^2)$.

\end{proof}

\subsection{Beyond the local expansion: can Fisher information control Rényi divergence globally?}
\label{app:fisher-global-counterexample}

A natural question is whether this local relationship between Fisher information and Rényi divergences extends beyond infinitesimal perturbations. More precisely, one may ask if there is a \emph{distribution-free} functional $F$ such that, given an arbitrary regular family of distribution $P$ and any $\theta < \theta' \in \Theta$:  \[D_\alpha(P_{\theta'}, P_\theta) = F(\alpha, \{I(u); \;u \in [\theta,\theta']\}).\] In general, no such characterization exists: Fisher information is a local quantity, and it does not determine Rényi divergence for arbitrary parameter values. Our counterexample of Theorem 4.1 of~\citet{Abbasnejad2006}, studied below, is a simple illustration of two family distributions that can have the same Fisher information but different Rényi divergences. There are, however, specific families of distributions in which global formulas are available. First, for distributions belonging to the same exponential family, Fisher information~\citep[Theorem~5.4]{Lehmann1998} and Rényi divergences~\citep[Chapter~2]{Liese1987} admit a closed form. Specific choices of parameters yield an integral representation of Rényi divergences in terms of the Fisher information~\citep{Bercher2012}, but these parametrizations do not simplify the theoretical calculation of the associated Rényi divergences, as the ratio $\frac{p_{\theta'}(x)^t}{p_{\theta}(x)^{1-t}}$ $(t \in (0,\alpha))$ appear in the parametrized densities.

A different line of work attempts to obtain global comparison criteria. For instance, one may hope that pointwise ordering of Fisher information implies ordering of Rényi divergences:
\[\text{ For all }\; \theta \in \Theta,\; I_P(\theta) \leq I_Q(\theta) \implies \text{For all } \theta < \theta' \in \Theta,\; D_\alpha(P_{\theta'},P_\theta) \leq D_\alpha(Q_{\theta'},Q_\theta),\]
where $P$ and $Q$ are two families of distributions indexed by the same parameter set $\Theta$, and the corresponding Fisher information are $I_P$ and $I_Q$.

Two existing attempts in the literature for Rényi divergences~\citep[Theorem~4.1]{Abbasnejad2006} and Kullback-Leibler (KL) divergences~\citep[Theorem~3.1]{Habibi2006} obtain such a statement. However, they are false in general, even for smooth regular families. We provide a counterexample and discuss the flaw in the proofs below.

\begin{theorem}[Adapted from Theorem 4.1 of~\citet{Abbasnejad2006}]
\label{thm:renyi-fisher-comparison-wrong}
Let $P = \{P_\theta,\ \theta\in \Theta\}$ and $Q = \{Q_\theta,\ \theta\in \Theta\}$ be two families of distributions. Let $0 < \theta_0 < \theta_1 \in \Theta$. Let $\alpha > 1$. Assume that both families have finite Rényi divergence and Fisher information, continuous in $\theta$. If the following conditions hold:
\begin{enumerate}
    \item \label{item:condition-1} For all $\theta \in I=[\theta_0,\theta_1], \;I_P(\theta) - I_Q(\theta) \geq d_0 > 0$,
    \item \label{item:condition-2} For all $\theta \in (\theta_0,\theta_1)$, $\delta \mapsto D_\alpha(p_{\theta + \delta},p_\theta)$ and $\delta \mapsto D_\alpha(q_{\theta + \delta},q_\theta)$ are three times differentiable and the third derivatives are uniformly bounded in a neighborhood $I = [0,c]$ of zero, then:
\end{enumerate}
\[ D_\alpha(p_{\theta_1},p_{\theta_0}) >D_\alpha(q_{\theta_1},q_{\theta_0}).\]
\end{theorem}
We start by pointing a counterexample to this theorem.

\paragraph{Counterexample of Theorem 4.1 from~\citet{Abbasnejad2006}:}
Fix $\alpha = 2$. Let $a, \sigma > 0$. Let $P_a = \{ \Cauchy(\theta, a),\; \theta \in \mathbb{R}\}$, where $\theta$ is the location parameter of the Cauchy distribution and $a$ is the scale parameter, and $Q_\sigma = \{ \mathcal{N}(\theta, \sigma^2),\; \theta \in \mathbb{R}\}$. Let $\theta_0 < \theta < \theta_1 \in \mathbb{R}$. Then, for the location parameter the Fisher information is constant for both families:
\[I_{P_a}(\theta) = \frac{1}{2 a^2},\;\;\; I_{Q_\sigma}(\theta) = \frac{1}{\sigma^2}.\]
 Also, the Rényi divergence between Cauchy distributions~\citep{Verdu2023} and between Gaussian distributions is known:
 \[D_2(p_{\theta_1}, p_{\theta_0}) = \log\left(1 + \frac{(\theta_1-\theta_0)^2}{2 a^2}\right),\;\;\; D_2(q_{\theta_1}, q_{\theta_0}) = \frac{(\theta_1-\theta_0)^2}{\sigma^2}.\]
 Both families of distributions satisfy the condition~\ref{item:condition-2} of Theorem~\ref{thm:renyi-fisher-comparison-wrong}. Moreover, choosing $a < \frac{1}{\sqrt{2}}\sigma$ allows to satisfy the condition~\ref{item:condition-1}: for all $\theta \in [\theta_0,\theta_1]$, $I_{Q_\sigma}(\theta) \leq I_{P_a}(\theta)$.
 However, for sufficiently large values of $\theta_1 - \theta_0$, it is clear that $D_2(p_{\theta_1}, p_{\theta_0}) \leq D_2(q_{\theta_1}, q_{\theta_0})$.

\paragraph{Flaw in the proofs of Theorem 4.1 from~\citet{Abbasnejad2006} and Theorem 3.1 from~\citet{Habibi2006}:}
The core idea of~\citet{Abbasnejad2006} and~\citet{Habibi2006} relies on showing that, under condition~\ref{item:condition-1}, for all $\theta \in \Theta$, there exists $\delta, d_1 > 0$ such that $\frac{1}{\delta^2}(D_\alpha(p_{\theta + \delta}, p_\theta) - D_\alpha(q_{\theta + \delta}, q_\theta)) \geq d_1$. Then, they incorrectly drop the $1/\delta^2$, by stating that when $c < 1$ (so that $\delta < 1$), $D_\alpha(p_{\theta + \delta}, p_\theta) - D_\alpha(q_{\theta + \delta}, q_\theta) \geq d_1$, which is generally false. The proof is not fixable independently of the regularity conditions. In fact, they try to generalize this statement for larger values of $\delta$ by induction to derive a lower bound on $D_\alpha(p_{\theta + k\delta}, p_\theta) - D_\alpha(q_{\theta + k\delta}, q_\theta)$ for $k > 1$, by a telescoping argument, but now that the lower bound of $D_\alpha(p_{\theta + \delta}, p_\theta) - D_\alpha(q_{\theta + \delta}, q_\theta)$ depends on $\delta$, which depends on $\theta$, the recurrence fails.

The theorem of~\citet{Habibi2006} for KL divergence follows the same proof technique and is also false.

\subsection{A sufficient condition to control Rényi divergence with Fisher information (Proposition~\ref{prop:Fisher-upper-bound})}
\label{app:prop-Fisher-upper-bound}

In this section, we introduce a criterion that allows to obtain an upper bound of Rényi divergences as a function of Fisher information along a parameter segment $[\theta,\theta']$. The derivation assumes an explicit upper envelope on the Rényi divergence that is available in our setting. It has the desirable property that it vanishes when the Fisher information becomes small on $[\theta,\theta']$ (provided $D_{2\alpha-1}$ remains controlled). This upper bound is generally not tight, as it scales as $\sqrt{I(\theta)}$ when the Fisher information is small, while the the local expansion of Theorem~\ref{thm:local-renyi} suggests a $I(\theta)$ rate. Here is the statement of the criterion:

\begin{proposition}[Upper bounding Rényi divergences through Fisher information]
\label{prop:app-renyi-criterion}
    Let $P :=\{P_\theta, \theta\in \Theta \subset \mathbb{R}\}$ be a regular family of distributions (Definition~\ref{defi:regularity}). Let $\theta < \theta' \in \Theta$.  We assume that there exists a function $U : \mathbb{R}^2 \to \mathbb{R^+}$ such that for all  $z \in (\theta,\theta')$, $D_{2\alpha-1}(P_{z},P_{\theta}) \leq U(z,\theta)$. Then,
    \[D_\alpha(P_{\theta'}, P_\theta) \leq \frac{1}{\alpha-1}\log \left( 1+ \alpha\sup_{z \in (\theta,\theta')} \sqrt{I(z)} \int_{\theta}^{\theta'} e^{(\alpha-1)U(z,\theta)}dz\right).\]
\end{proposition}

\begin{proof}
    For $\theta,\theta' \in \Theta$, we note $H(\theta',\theta) = \mathbb{E}_{Y_\theta \sim P_\theta}\left[r_{\theta', \theta}^\alpha\right]$. Fix $\theta \in \Theta$ and note for every $\theta' \in \Theta$, $G(\theta') = H(p_{\theta'}, p_{\theta})$. Then, for any $h > 0$ such that $\theta'+h \in \Theta$,  $r_{\theta' + h, \theta} = r_{\theta'+h, \theta'}  r_{\theta', \theta}$. Therefore, $\dot r_{\theta',\theta} = s_{\theta'} r_{\theta',\theta}$. Since $P$ is regular, the conditions of Definition~\ref{defi:regularity} ensure differentiability of $r_{\theta',\theta}$ under the expectation. Then, $G'(\theta') = \alpha\mathbb{E}_{P_\theta}[s_{\theta'} r_{\theta',\theta}^{\alpha}]$.

    Cauchy-Schwarz gives:
    
    \begin{align*}
        G'(\theta') =  \alpha\mathbb{E}_{P_{\theta'}}[s_{\theta'} r_{\theta',\theta}^{\alpha-1}] &\leq  \alpha\mathbb{E}_{P_{\theta'}}\left[s_{\theta'}^2\right]^{1/2} \mathbb{E}_{P_{\theta'}}\left[r_{\theta',\theta}^{2(\alpha-1)}\right]^{1/2} \\ &= \alpha I(\theta')^{1/2} e^{(\alpha-1)D_{2\alpha-1}(P_{\theta'},P_{\theta})}\\ &\leq \alpha \sup_{z \in (\theta,\theta')} I(z)^{1/2} e^{(\alpha-1)U(\theta',\theta)}.
    \end{align*}

    Integrating over $(\theta,\theta')$, we get:

    \[G(\theta') = G(\theta) + \int_{\theta}^{\theta'} \partial_{z} G(z) dz \leq 1+ \alpha\sup_{z \in (\theta,\theta')} \sqrt{I(z)} \int_{\theta}^{\theta'} e^{(\alpha-1)U(z ,\theta)}dz,\]
    giving the desired result.

\end{proof}

This upper bound can be directly leveraged in our setting. In fact, the Rényi divergence $D_{2\alpha-1}(ZV,ZW)$ is upper bounded by post-processing: $D_{2\alpha-1}(V^\top V,W^\top W) \leq D_{2\alpha-1}(V,W)$. Therefore, for any choice of parametrization of $V^\top V$, we can obtain an upper bound of $D_{2\alpha-1}(V^\top V, W^\top W)$ as a function of the post-processing upper bound and the Fisher information related to this parametrization.

As a validation step, we can compare this upper bound to the exact Rényi divergence in the case of the Rényi divergence between Gaussian distributions: $D_\alpha(N+z,N+w) = \frac{\alpha}{2}(z-w)^2$, and the Fisher information is $I(z) = 1$.
Then, 
\begin{align*}
    D_\alpha(N+v,N+w) &\leq \frac{1}{\alpha-1} \log \left( 1+ \frac{\alpha}{} \int e^{(\alpha-1)(2\alpha-1) \frac{(z-w)^2}{2}}dz\right) \\
    &\leq \frac{1}{\alpha-1} \log \left( 1+ \frac{2\alpha}{(\alpha-1)(2\alpha-1)(v-w)} \left(e^{(\alpha-1)(2\alpha-1) \frac{(v-w)^2}{2}}-1\right)\right)
\end{align*}
Figure~\ref{fig:tightness-fisher-gaussian} represent the upper bound and the original divergence in the case $\alpha=2$, as a function of $\Delta$.
\begin{figure}[h]
    \centering
    \includegraphics[width=0.6\linewidth]{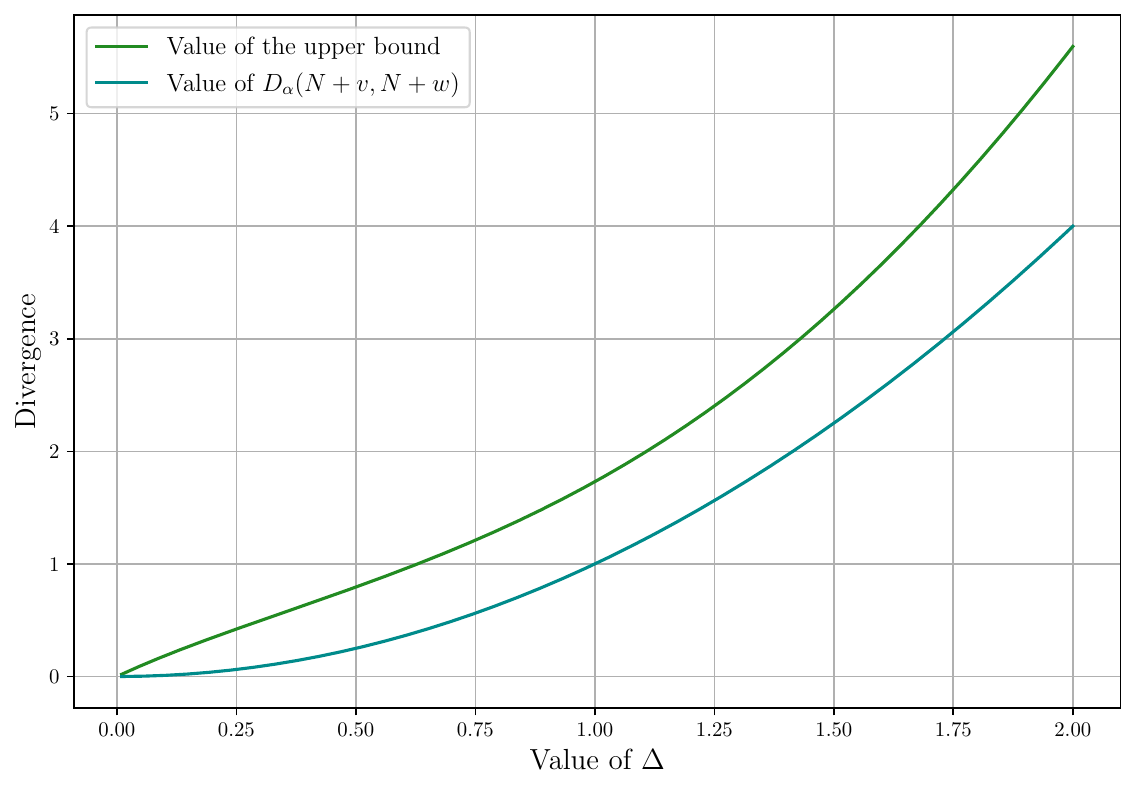}
    \caption{Comparison between the Rényi divergence $D_\alpha(N+v,N+w)$ and the criterion obtained from Proposition~\ref{prop:Fisher-upper-bound} as a function of $\Delta$, in the case $\alpha = 2$.}
    \label{fig:tightness-fisher-gaussian}
\end{figure}

\newpage

\section{From Outputs to Sufficient Gram Statistics}

\subsection{Proof of Proposition~\ref{prop:sufficient-statistics}}
\label{app:prop-sufficient-statistics}
\propSufficientStats*

\begin{proof}
     $ZV$ is invariant under left multiplication by any orthogonal matrix $U \in O(n_{\mathrm{syn}})$. Therefore, for any $U, U' \in O(n_{\mathrm{syn}})$, $\mathcal{L}(UZV)= \mathcal{L}(U'ZV)$. We note $T(ZV) = V^\top Z^\top Z V$. Note $s \in \mathbb{S}^{ k}$ be a symmetric matrix such that $\rank(s) \leq n_{\mathrm{syn}}$. Note $R(s) \in \mathbb{R}^{n_{\mathrm{syn}} \times k}$ such that $R(s)^\top R(s) = s$. Then, $T(U R(s)) = R(s)^\top U^\top U R(s) = s$ for any $U \in O(n_{\mathrm{syn}})$. Define the Markov kernel $K(s,\cdot) = \mathcal{L}(UR(s))$, where $U \sim \mathrm{Unif}(O(n_{\mathrm{syn}}))$. Then, by left orthogonal invariance, the conditional distribution of $ZV$ conditioned on $T(ZV) = s$ is $K(s,\cdot)$. Then, we have: $\mathcal{L}(ZV) = K \mathcal{L}(T(ZV))$.
     Therefore, by post-processing, \[D_\alpha(ZV,ZW) = D_\alpha(K \mathcal{L}(T(ZV)), K \mathcal{L}(T(ZW))) \leq D_\alpha(V^\top Z^\top Z V, W^\top Z^\top Z W).\]
     The reverse direction post-processing inequality yields the result:
     \[D_\alpha(T(ZV), T(ZW)) \leq D_\alpha(ZV,ZW).\]
\end{proof}

\subsection{Proof of Proposition~\ref{prop:gram-limit}}
\label{app:prop-gram-limit}

\propGramLimit*

\begin{proof}
First, observe that \[\frac{1}{n_{\text{syn}}}Z^\top Z = \frac{1}{n_{\text{syn}}}\sum_{i=1}^{n_{\text{syn}}} Z_i^\top Z_i \overset{n_{\text{syn}} \to \infty}{\to} I_d,\]
by the law of large numbers. This means that $\frac{1}{n_{\text{syn}}}V^\top Z^\top ZV \to V^\top V, \frac{1}{n_{\text{syn}}} W^\top Z^\top Z W \to W^\top W $.
    Also, Rényi divergences are lower semi-continuous for the weak topology~\citep{vanErven2014}.
    This means that:
    \begin{align*}
        \underset{n_{\text{syn}} \to \infty}{\lim} \inf D_\alpha(V^\top Z^\top Z V, W^\top Z^\top Z W) &= \underset{n_{\text{syn}} \to \infty}{\lim} \inf D_\alpha\left(\frac{1}{n_{\text{syn}}}V^\top Z^\top Z V, \frac{1}{n_{\text{syn}}}W^\top Z^\top Z W\right)\\ &\geq D_\alpha(V^\top V, W^\top W).
    \end{align*}

    Also, we define the kernel $K : \Sigma \mapsto \mathcal{W}(k,n_{\text{syn}},\Sigma)$, for $\Sigma \in \mathbb{S}_k^+(\mathbb{R})$, or equivalently: $K(\Sigma) = \Sigma^{1/2} G^\top G \Sigma^{1/2}$.
    
    Now, we note $\Sigma_V = V^\top V$. Conditionally on $V$, each row of $ZV$ is Gaussian with covariance $V^\top V$. Therefore, for a standard Gaussian matrix $G \in \mathbb{R}^{\mathrm{syn} \times k}$, $ZV  | V \overset{d}{=} G \Sigma_V^{1/2}$. Therefore, $(ZV)^\top ZV | V \overset{d}{=} \Sigma_V^{1/2} G^\top G \Sigma^{1/2}$. By scaling invariance of Rényi divergences,
    \begin{align*}
        D_\alpha(V^\top V, W^\top W) \geq D_\alpha(K(V^\top V), K(W^\top W)) &= D_\alpha\left(\Sigma_V^{1/2}G^\top G \Sigma_V^{1/2}, \Sigma_W^{1/2}G^\top G \Sigma_W^{1/2} \right)\\
        &= D_\alpha\left(V^\top Z^\top Z V, W^\top Z^\top Z W \right).
    \end{align*}
    Finally,
\[D_\alpha(V^\top Z^\top Z V, W^\top Z^\top Z W)\leq  D_\alpha(V^\top V, W^\top W) \leq \underset{n_{\text{syn}} \to \infty}{\lim} \inf D_\alpha(V^\top Z^\top Z V, W^\top Z^\top Z W),\]
giving the desired result.
\end{proof}

\subsection{Proof of Proposition~\ref{prop:nofreelunch}}
\label{app:prop-nofreelunch}

\propNoFreeLunch*

\begin{proof}
    The idea of the proof is to take $v_\theta = t e_1, w_t = (t + \Delta)e_1$, where $e_1 = (\delta_{11}) \in \mathbb{R}^{d\times k}$. We note $V_t = N + v_t$ and $W_t = N + w_t$.
    First, by post-processing:
    \[D_\alpha(V_t^\top V_t,W_t^\top W_t) \geq D_\alpha((V_t^\top V_t)_{11}, (W_t^\top W_t)_{11})) = D_\alpha(\chi_d^2(t^2),\chi_d^2((t+\Delta)^2)).\]
    Now, take $f : x \mapsto \sqrt{|x|}$. As $t \to \infty$, $f$ becomes injective on a set of probability $p_t \to 1$. The distribution of $f(\chi_d^2(t^2))$ converges weakly to $\mathcal{N}(t,1)$. Then, for all $t > 0$:
    \[\underset{t \to \infty}{\lim} \inf D_\alpha(f(\chi_d^2(t^2)),f(\chi_d^2((t + \Delta)^2)) \geq D_\alpha(\mathcal{N}(t,1), \mathcal{N}(t+\Delta, 1)) = D_\alpha(V,W),\]
    which is finite, giving the desired result.
\end{proof}

\section{Privacy amplification in linear generators}
\label{app:fisher-one-dim}

In this section, we restate the introduction of Section~\ref{sec:privacy-amp-linear} and add relevant results and details. There is, to our knowledge, no existing analysis of the Fisher information of non-central Wishart distributions for the non-centrality parameter. The Fisher information of non-central chi-squared distributions ($k=1$) for the non-centrality parameter is not available in closed form~\citep{Idier2014}. As a result, much of the literature has focused either on deriving lower bounds~\citep{Bouhrara2018, Idier2014, Stein2016} or on numerical estimation methods~\citep{Bouhrara2018}. The literature on upper bounds remains limited, aside from the trivial post-processing inequality $I(\theta)[H] \leq 1$.

In their work,~\cite{Idier2014} focus on the Fisher information of Rician and generalized Rician distributions with respect to the amplitude parameter $\theta$. \citet{Idier2014} find the following upper bound:
\[I_X(d,\theta) \leq \min\left\{1,\frac{2\theta^2}{d}\right\}.\]
While useful in the $d \gg \theta^2$ regime, this bound is piecewise and does not naturally capture the asymptotic behaviour $I_X(d,\theta) \overset{\theta \to +\infty}{\to} 1$.

Aside from the work of~\citet{Idier2014} which proposes an upper bound for the Fisher information of generalized Rician distributions. All these works only focus on the $k=1$ case, while the $k>1$ case is more convoluted. In fact, although the diagonal elements of $V^\top V$ are independent non-central chi-squared random variables, the non-diagonal elements are not independent, and follow a substantially more complex joint distribution.
  
\subsection{Proof of Theorem~\ref{thm:Fisher-upper-bound}}

In order to prove Theorem~\ref{thm:Fisher-upper-bound}, we leverage the following lemmas:
\begin{lemma}
\label{lem:parceval-hermite}
    Let $F \in L^2(\gamma)$ be a centered variable: $\mathbb{E}[F] = 0$. We note, for $N \sim \gamma$, $b_F = \mathbb{E}[NF]$.
    Define the symmetric linear operator $M_F$ as the second Hermite coefficient of $F$:
    \[\langle B, M_F C\rangle_F = \mathbb{E}[F(\langle B,N\rangle_F \langle C,N\rangle_F - \langle B, C\rangle_F)].\]

    Then,
    \[\|M_F\|_{\rm HS}^2 \leq 2 (\mathbb{E}[F^2] - \|b_F\|_F^2).\]
\end{lemma}

This is a standard inequality.

\begin{lemma}
\label{lem:hermite-identity}
    Assume that $F(N) = \phi((N+v)^\top(N+v))$. Let $A = - A^\top$ be a skew-symmetric matrix. Then,
    \[M_F(Av) = A b_F.\]
\end{lemma}

\begin{proof}
We let $A$ be a skew symmetric matrix: $A^\top = -A$. We show that the orthogonal invariance provided by the Gram matrix operation enables to show $M_F(Av) = Ab_F$. To prove that, we define $R_t = e^{tA}$ and $T_t z = R_t(z+v)-v$, and leverage that non-central Wishart distributions are invariant by $T_t$. Since $F$ is $\sigma(Y)$-measurable, there exists a measurable function $\phi$ such that $F = \psi((N+v)^\top(N+v)) = \phi(N)$. Then,
\[\phi(T_t N) = \psi((T_t N + v)^\top (T_t N + v)= \psi( (N+v)^\top R_t^\top R_t (N+v)) = \psi((N+v)^\top (N+v)) = \phi(N).\]

Let $U \in \mathbb{R}^{d \times k}$ be a matrix. Then,
\[\mathbb{E}[\phi(T_t N) \langle U, N\rangle_F] = \mathbb{E}[\phi(N) \langle U, N\rangle_F].\]
The map $T_t$ is an affine bijection with inverse $T_{-t}$. Since $R_t$ is orthogonal and $\det R_t = 1$, the Jacobian is $1$. Therefore for any integrable test function $h$, the change of variables $z \mapsto T_t z$ gives $\mathbb{E}[h(T_tN)] = \mathbb{E}[h(N)\rho_t(N)]$, with $\rho_t(N) = \exp\left(-\frac{1}{2}\|T_{-t} N\|_F^2 + \frac{1}{2}\|N\|_F^2\right)$.
Then,
\[\mathbb{E}[\phi(N) \langle U, N\rangle_F] = \mathbb{E}[\phi(N) \langle U, T_{-t}N\rangle_F \rho_t(N)].\]
Then, $\mathbb{E}[\phi(N) \langle U, T_{-t}N\rangle_F \rho_t(N)]$ does not depend on $t$. We can differentiate at $t=0$, since the integrand is differentiable and $T_{-t} N$ is affine in $N$ up to smooth coefficients, $\rho_t(N)$ is the exponential of an affine function of $N$. We obtain:
\begin{align*}
    0 &= \mathbb{E}\left[\phi(N) \rho_t(N) (\langle U, A e^{-tA} (N+v)\rangle_F - \langle U, T_{-t} N\rangle_F \langle T_{-t} N, A(T_{-t} N + v)\rangle_F)\right]\\
    &= \mathbb{E}\left[\phi(N)  (\langle U, A  (N+v)\rangle_F - \langle U,  N\rangle_F \langle N, A( N + v)\rangle_F)\right].
\end{align*}

Therefore, since $A^\top = -A$, $\langle N, AN\rangle_F = 0$, and
\[\mathbb{E}\left[\phi(N)  (\langle U, A  v\rangle_F - \langle U,  N\rangle_F  \langle N, A v \rangle_F)\right]   = -\mathbb{E}[\phi(N) \langle U, AN\rangle_F].\]
We identify:
\[\langle U, M_F(Av)\rangle_F = \mathbb{E}[F(\langle Av,N\rangle_F \langle U,N\rangle_F - \langle Av, U\rangle_F)].\]

Therefore, we obtain:
\[\langle U, M_F(Av)\rangle_F  = \mathbb{E}[F \langle U, AN\rangle_F] = \langle U, A\mathbb{E}[NF]\rangle_F = \langle U, A b_F\rangle_F.\]
Since this equality holds for all $U \in \mathbb{R}^{d \times k}$, we have the desired result: \[M_F(Av) = A b_F.\]
\end{proof}

Leveraging these two lemmas, we can now prove the Fisher information upper bound:

\begin{proposition}
\label{prop:hermite-bound}
    Let $F \in L^2(\sigma(Y))$ such that $\mathbb{E}[F] = 0$. We have:
    \[\|b_F\|_F^2 \leq \frac{2 \|v\|^2}{2\|v\|^2 + d - 1} \mathbb{E}[F^2].\]
\end{proposition}

\begin{proof}
    Let $A_{pq} = E_{pq} - E_{qp}$, with $1 \leq q < p \leq d$ be skew-symmetric matrices.
    
    Then, by Lemmaß\ref{lem:hermite-identity}, we have: $M_F(A_{pq} v) = A_{pq} b_F$.

    Therefore, by summing over all the combinations,

    \begin{align*}
        \sum_{1 \leq p < q \leq d}\|M_F(A_{pq} v)\|_{\mathrm{HS}}^2 = \sum_{1 \leq p < q \leq d} \|A_{pq} b_F\|_F^2.
    \end{align*}

 Also,
    \[\sum_{1 \leq p < q \leq d} \|A_{pq} \mathbb{E}[NF]\|_F^2 = \sum_{1 \leq p < q \leq d} \|\mathbb{E}[NF]_{p,\cdot}\|_2^2 + \|\mathbb{E}[NF]_{q,\cdot}\|_2^2.\]
    Each row appears exactly $d-1$ times in the sum. Therefore, \[\sum_{1 \leq p < q \leq d} \|A_{pq} \mathbb{E}[NF]\|_F^2 = (d-1) \|\mathbb{E}[NF]\|_F^2.\]

    Also, for all $v,w \in \mathbb{R}^{d\times k}$, $\sum_{1 \leq p < q \leq d}\langle w, A_{pq} v\rangle_F^2 \leq \|w\|_F^2\|v\|_F^2$. In fact,

    \begin{align*}
         \sum_{1 \leq p < q \leq d}\langle w, A_{pq} v\rangle_F^2 &= \sum_{1 \leq p < q \leq d}\left( \sum_{i=1}^d\sum_{j=1}^k\sum_{l=1}^d w_{ij} (A_{pq})_{il} v_{lj}\right)^2\\
         &= \sum_{1 \leq p < q \leq d}\left( \sum_{i=1}^d\sum_{j=1}^k\sum_{l=1}^d w_{ij} (\delta_{ip}\delta_{lq} -\delta_{iq}\delta_{lp} )v_{lj}\right)^2\\
         &= \sum_{1 \leq p < q \leq d}\left( \sum_{j=1}^k w_{pj} v_{qj} - w_{qj} v_{pj}\right)^2\\
         &= \sum_{1 \leq p < q \leq d}\left( \langle w_p, v_q\rangle - \langle v_p, w_q\rangle\right)^2\\
    \end{align*}
    Then, by Minkowski inequality,
    \[\left(\sum_{1 \leq p < q \leq d}\left( \langle w_p, v_q\rangle - \langle v_p, w_q\rangle\right)^2\right)^{1/2} \leq \sum_{j=1}^k \left( \sum_{1 \leq p < q \leq d} (w_{pj} v_{qj} - w_{qj} v_{pj})^2\right)^{1/2}.\]
    By Lagrange's identity:
    \[\sum_{1 \leq p < q \leq d} (w_{pj} v_{qj} - w_{qj} v_{pj})^2 = \|w_{\cdot, j} \|^2 \|v_{\cdot, j} \|^2 - \langle w_{\cdot, j} , v_{\cdot, j}\rangle^2\]

    Therefore, 
    \begin{align*}
        \left(\sum_{1 \leq p < q \leq d}\left( \langle w_p, v_q\rangle - \langle v_p, w_q\rangle\right)^2\right)^{1/2} &\leq \sum_{j=1}^k \left(\|w_{\cdot, j} \|^2 \|v_{\cdot, j} \|^2 - \langle w_{\cdot, j} , v_{\cdot, j}\rangle^2\right)^{1/2}\\
        &\leq  \sum_{j=1}^k \|w_{\cdot, j} \| \|v_{\cdot, j} \| \leq \|w\|_F \|v\|_F.
    \end{align*}

    Now, let $(e_i)$ be a basis of $\mathbb{R}^{d \times k}$. By symmetry of $M_F$:
    \begin{align*}
        \sum_{1 \leq p < q \leq d}\|M_F(A_{pq} v)\|_{\mathrm{HS}}^2 &= \sum_{j=1}^{dk}\sum_{1 \leq p < q \leq d}\langle e_j, M_F(A_{pq} v)\rangle_F^2 \\
        &= \sum_{j=1}^{dk}\sum_{1 \leq p < q \leq d}\langle M_F^*(e_j), A_{pq} v\rangle_F^2 \\
        &\leq \sum_{j=1}^{dk}\|M_F^* (e_j)\|_{F}^2 \|
        v\|_F^2\\
        &\leq \|M_F\|_{\mathrm{HS}}^2 \|v\|_F^2.
    \end{align*}
    Leveraging Lemma~\ref{lem:parceval-hermite}, we obtain:
    \begin{align*}
        (d-1) \|b_F\|_F^2 \leq \|M_F\|_F^2 \|v\|_F^2 \leq 2 \|v\|_F^2 (\mathbb{E}[F^2] - \|b_F\|_F^2).
    \end{align*}

    Finally,
    \[\|b_F\|_F^2 \leq \frac{2 \|v\|_F^2}{2 \|v\|_F^2 + d-1} \mathbb{E}[F^2].\]
\end{proof}

Now, we can prove Theorem~\ref{thm:Fisher-upper-bound}.

\thmUpperBoundFisher*

\begin{proof}
    We leverage our variational representation of Fisher information:
    \[I(\theta)[H] = \sup_{F \in L^2(\sigma(Y)), \mathbb{E}[F] = 0, \mathbb{E}[F^2] > 0} \frac{\mathbb{E}[F \langle H, N \rangle_F]^2}{\mathbb{E}[F^2]} \]
    Then, Let $F \in L^2(\sigma(Y))$ such that $ \mathbb{E}[F] = 0$ and $ \mathbb{E}[F^2] > 0$. Proposition~\ref{prop:hermite-bound} shows that:
    \[\|\mathbb{E}[NF]\|_F^2 \leq \frac{2 \|v\|_F^2}{2 \|v\|_F^2 + d-1} \mathbb{E}[F^2].\]
    Moreover, by the Cauchy-Schwarz inequality, $\mathbb{E}[F \langle H, N\rangle_F]^2 \leq \|H\|_F^2 \|\mathbb{E}[FN]\|_F^2$, yielding the result.
\end{proof}

\subsection{Proof of Theorem~\ref{thm:local-renyi}}

\thmUpperBoundRenyiLocal*

\begin{proof}
    Let $v, w \in \mathbb{R}^{d \times k}$. If $v = w$, the claim is trivial. Assume that $\Delta := \|v-w\|_F > 0$ and define:
    \[u := \frac{v-w}{\Delta}, \qquad \theta_s := w + su, \qquad s\in [0,\Delta].\]
    Let $P_s := \mathcal{L}((N+\theta_s)^\top (N+\theta_s))$. Then, $P_0 = \mathcal{L}(W^\top W)$ and $P_{\Delta} = \mathcal{L}(V^\top V)$. The family $\{P_s,\; s \in [0,\Delta]\}$ is a one-dimensional parametric family. Its scalar Fisher information is $\tilde I(s) := \mathbb{E}_{Y \sim P_s} [(\partial_{s'} (dP_{s'}/dP_s)|_{s' = s}(Y))^2]$, and by the chain rule, $\tilde I (s) = I(\theta_s)[u]$, where $I(\theta)[H]$ denotes the directional Fisher information of the non-central family. Applying Proposition~\ref{prop:local-fisher} and leveraging~\ref{thm:Fisher-upper-bound}, we obtain:
    \[D_\alpha(V^\top V, W^\top W) \leq \frac{\alpha}{2} \left( \frac{2 R(v,w)^2}{2 R(v,w)^2 + d - 1} + o(1)\right) \Delta^2,\]
    giving the desired result.
    
    Note that this upper bound can be optimized by leveraging $\inf_{U \in O(d)} \|Uv-w\|_F$ instead of $\|v-w\|_F$.

\end{proof}

\begin{remark}
    In the case $k=1$,~\cite{Idier2014} give the following lower bound on Fisher information, for the amplitude parameter $\theta$:
    \[I(\theta) \geq \frac{2 \|\theta\|^2}{2\|\theta\|^2 + d}.\] Therefore, we can also extend our results to a lower local envelope on the Rényi divergence $D_\alpha(V^\top V, W^\top W)$ for the case $k=1$:
    \[D_\alpha(V^\top V, W^\top W) \geq \frac{\alpha}{2}\frac{2 \min\{\|v\|, \|w\|\}^2}{2\min\{\|v\|, \|w\|\}^2 + d} (\|v\| - \|w\|)^2 + o((\|v\| - \|w\|)^2).\]
\end{remark}

\newpage

\subsection{Proof of Theorem~\ref{thm:upper-bound-renyi}}

Before proving the theorem, here is a sharp elementary upper bound to $\int_0^x e^{a t^2} dt$:
\begin{lemma}
\label{lem:gauss-simple-bound}
    Let $a, x > 0$. Then:
    \[\int_0^x e^{a t^2} dt \leq \frac{e^{a x^2}-  1}{a x}.\]
\end{lemma}

\thmUpperBoundRenyi*

\begin{proof}
We use the same parametrization strategy $P_s$ as in the proof of Theorem~\ref{thm:local-renyi}. We apply Proposition~\ref{prop:Fisher-upper-bound} to the parametric family $\{P_s, \; s \in [0,\Delta]\}$, with $\Delta = \|v-w\|_F$. First, by post-processing:
\[D_{2\alpha-1}(P_s, P_0) \leq D_{2\alpha-1}(N+\theta_s,N+w) = (2\alpha-1) \frac{s^2}{2}.\]
Thus, we may take:
\[U(s,0) = (2\alpha-1) \frac{s^2}{2}.\]
By Theorem~\ref{thm:Fisher-upper-bound},
\[\tilde I(s) \leq \frac{2\|\theta_s\|^2}{2\|\theta_s\|^2 + d-1}\|v-w\|_F^2/\Delta^2.\]
Leveraging~\ref{prop:Fisher-upper-bound}, we obtain:
    \[D_\alpha(V^\top V,W^\top W) \leq \frac{1}{\alpha-1}\log \left( 1 + \alpha\sup_{s \in (0,\Delta)} \sqrt{\frac{2\|\theta_s\|^2}{2\|\theta_s\|^2 + d-1}}\|v-w\|_F\int_0^\Delta e^{(\alpha-1)(2\alpha-1)s^2/2}dz /\Delta\right).\]
Finally, leveraging Lemma~\ref{lem:gauss-simple-bound}, we find:
    \[ D_\alpha(V^\top V,W^\top W) \leq \frac{1}{\alpha-1}\log\left(1+2\alpha  \left(\frac{2R(v,w)^2}{2R(v,w)^2 + d - 1}\right)^{1/2}\frac{e^{(\alpha-1) (2\alpha-1) \frac{\Delta^2}{2}} -1}{(\alpha-1)(2\alpha-1)\Delta}\right),\]
    which proves the claim.
\end{proof}

\subsection{Interpretation of the Fisher comparison principle and Theorem~\ref{thm:upper-bound-renyi}}

Then, we can also represent the privacy guarantees obtained from Theorem~\ref{thm:upper-bound-renyi} as a function of $d$, as shown in Figure~\ref{fig:general-amplification-one-dim}, illustrating the privacy gains.

\begin{figure}[h]
    \centering
    \includegraphics[width=0.6\linewidth]{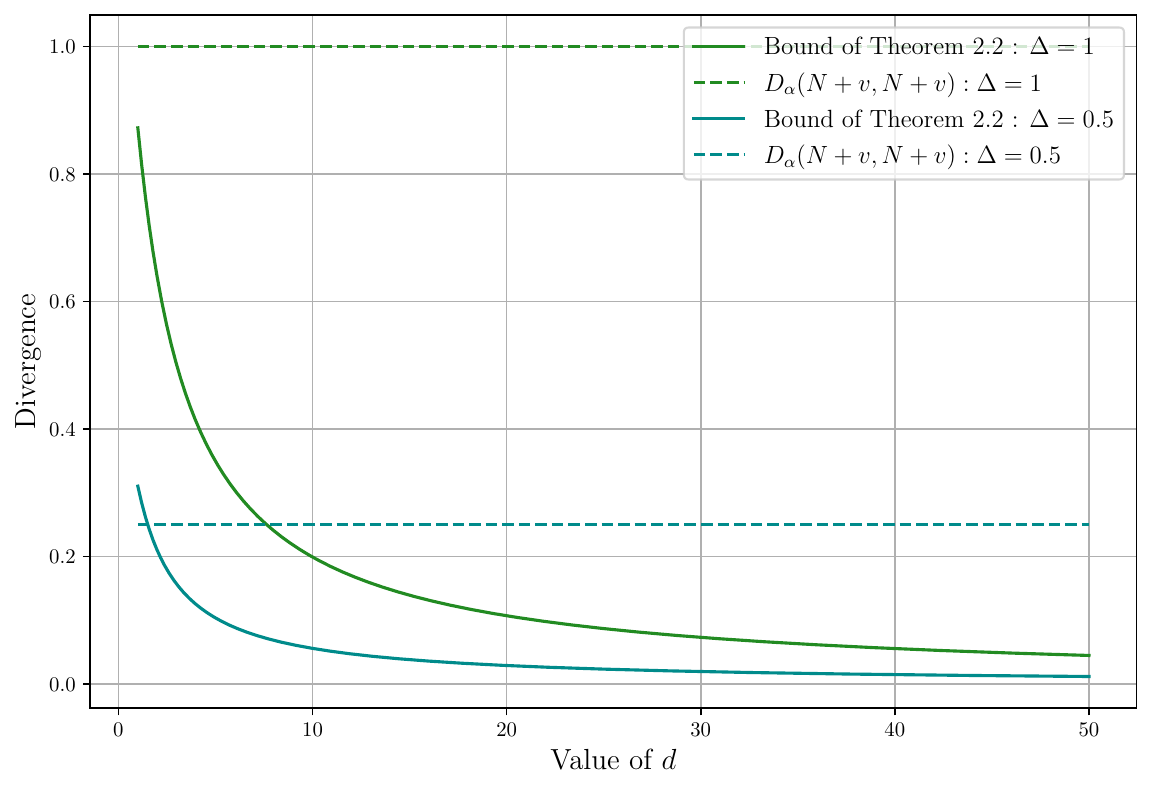}
    \caption{Privacy bounds of Theorem~\ref{thm:upper-bound-renyi} as a function of $d$ for multiple values of $\Delta$, $R(v,w) = 1$.}
    \label{fig:general-amplification-one-dim}
\end{figure}

\newpage

\newpage

\section{Numerical estimation of the divergence}
\label{app:experiments}
In this section, we restate the discussion of Section~\ref{sec:experiments} and add details for completeness.

We perform a series of experiments designed to illustrate our theoretical results. While sampling from $ZV$ is straightforward, its density is analytically intractable. Therefore, we rely on variational methods to estimate the Rényi divergence $D_\alpha(ZV,ZW)$. Specifically, we numerically estimate the Rényi divergence induced by releasing $ZV$ using the \emph{Convex-Conjugate Rényi Variational Formula}, introduced~\cite{birrell2023} (Theorem 2.1).

\begin{theorem}[Convex-Conjugate Rényi Variational Formula (Theorem 2.1 from~\citep{birrell2023})]
Let $(\Omega, \mathcal{M})$ be a measurable space, $P, Q \in \mathcal{P}(\Omega)$, and $\alpha \in (0,1) \cup (1,\infty)$. Then,
\[D_\alpha(P,Q) =\alpha \sup_{g \in \mathcal{M}_b(\Omega); g < 0} \left\{ \int gdQ + \frac{1}{\alpha - 1} \log \int |g|^{(\alpha-1)/\alpha} dP \right\} + \log \alpha + 1,\]
where $\mathcal{M}_b(\Omega)$ is the space of bounded functions on $\Omega$.
\end{theorem}

Among existing variational methods to estimate the Rényi divergence, this formula has the specificity of being well-behaved when $\alpha > 1$~\citep{birrell2023}, which aligns with our setting. In order to estimate the Rényi divergence between two distributions $P,Q$, we proceed as follows. We initialize a two layers neural network $f(w_0,\cdot)$ with negative poly-softplus activation~\citep{birrell2023} and train it to maximize:
\[\sup_{w} \left\{ \int f(w)dQ + \frac{1}{\alpha - 1} \log \int |f(w)|^{(\alpha-1)/\alpha} dP \right\}\]
Optimization is performed via minibatch stochastic gradient ascent. At each iteration, we sample $X = (X_1,\dots, X_l) \sim P^{\otimes l}$, $Y = (Y_1,\dots,Y_l) \sim Q^{\otimes l}$ for some batch size $l > 0$ and update the network parameters using the empirical objective:
\[F(w,X,Y) =\frac{1}{l}\sum_{i=1}^lf(w,Y_i)+  \frac{1}{
\alpha - 1}\log\left( -\frac{1}{l}\sum_{i=1}^l f(w,X_i)^{(\alpha-1)/\alpha}\right).\]
Each estimation is performed $10$ times and the results are averaged. For reproducibility, here are some experimental details:
\begin{itemize}
    \item Number of training steps: $2\cdot 10^4$.
    \item Batch size: $512$.
    \item Learning rate: $10^{-3}$.
    \item Hidden dimension: $64$.
    \item Final activation: Poly-softplus~\citep{birrell2023}.
    \item Optimizer: Adam~\citep{Kingma2014}.
    \item Validation score is computed every $100$ steps with a resampled validation set of size $5 \times 10^4$.
    \item Early stopping criterion: No validation score improvement for $10$ validation steps.
    \item Once trained, the divergence estimation is made on a sample of size $5 \times 10^4$.
\end{itemize}
At each step, data is resampled from the distributions of $ZV$ and $ZW$.

\subsection{Privacy Auditing on VAEs trained with DP-SGD.}

We perform the Querybased attack of~\cite{Annamalai2024} on their worst-case dataset setting. For reproducibility, here are the experimental details:

\begin{itemize}
    \item Adjacent datasets $\mathcal{D} = \{(1,1,1), (0,0,0), (0,0,0), (0,0,0)\}$ and $\mathcal{D}' = \{(1,1,1), (1,1,1), (0,0,0), (0,0,0)\}$
    \item VAE architecture:
    \begin{itemize}
        \item Number of layers: $2$.
        \item Hidden dimension: $64$.
        \item latent dimension: $\in \{4,8,16, 32, 64, 128\}$.
        \item Output dimension: $3$
    \end{itemize}
    \item Shadow models:
    \begin{itemize}
        \item Training: $3000$ models.
        \item Validation: $1000$ models.
        \item Test: $1000$ models.
    \end{itemize}
    \item Size of the synthetic datasets: $\in [10,1000]$.
    \item DP-SGD parameters:
    \begin{itemize}
        \item Noise multiplier: $1.0$.
        \item Batch fraction: $1$.
        \item Clipping parameter: $1$.
        \item Number of epochs: $50$
    \end{itemize}
    \item Learning rate: $2 \times 10^{-3}$
    \item VAE KL weight: $0.1$.
    \item Optimizer: Adam~\citep{Kingma2014}.
    \item Computation resources: $1\times$ GPU Nvidia A100. A round of experiments takes approximately $18$ hours.
\end{itemize}

\section{Discussion and extensions}

In this section, we discuss the tightness of the infinite data release privacy loss $D_\alpha(V^\top V, W^\top W)$ with respect to $D_\alpha(ZV,ZW)$. We also discuss potential directions for extending privacy amplification results to linear generators trained with widely used private learning algorithms.

\subsection{\texorpdfstring{On the tightness of the plateau $D_\alpha(V^\top V, W^\top W)$}{On the tightness of the plateau}}
\label{subsec:convergence-rate}
As discussed in Section~\ref{sec:experiments} and~\ref{sec:l_infty}, the privacy loss $D_\alpha(ZV,ZW)$ converges to $D_\alpha(V^\top V, W^\top W)$ when $n_{\text{syn}} \to +\infty$.
A natural question is whether the plateau $D_\alpha(V^\top V,W^\top W)$ is a meaningful approximation of $D_\alpha(ZV,ZW)$ for finite $n_{\text{syn}}$, especially when $d$ is large.
We show that, in the tractable case $k=1$, for sufficiently large $d$,  the rate of convergence of the difference $D_\alpha(V^\top V, W^\top W) - D_\alpha(ZV,ZW)$ to $0$ is \textbf{at least} as fast as $O(1/n_{\text{syn}})$, with constants that can be chosen uniformly over all large $d$.

\begin{proposition}[Privacy gap is $O(1/n_{\text{syn}})$ for $k=1$]
\label{prop:privacy-gap}
     Let $\alpha \geq 2$ and $v,w \in \mathbb{R}^{d}$. Let $\tau = (\|v\|,\|w\|,\Delta,\alpha)$ and assume that $k=1$. Then, there exists $ K = K(\tau), K' = K'(\tau)$ and a threshold $d_0 = d_0(\tau)$ such that for all $d \geq d_0$ and $n_{\text{syn}} \geq 1$:
    \[\textstyle D_\alpha(V^\top V, W^\top W) - D_\alpha(ZV,ZW) \leq  \frac{K}{n_{\text{syn}}} + K' e^{-c n_{\text{syn}}}.\]
\end{proposition}

We are interested in regimes where $D_\alpha(V^\top V,W^\top W)$ itself decreases with $d$ (e.g., scaling like $d^{-\gamma}$ for some $\gamma > 0$, $\gamma \geq 1/2$ follows from Theorem~\ref{thm:upper-bound-renyi}). In such regimes, the additive gap bound above can be used to assess tightness: achieving a constant-factor approximation of the plateau requires the gap to be smaller than the plateau scale.

\begin{corollary}[Constant factor tightness for large $n_{\text{syn}}$]
\label{corr:privacy-gap-scaling}
    Assume that $D_\alpha(V^\top V,W^\top W) = \Omega(d^{-\gamma})$, $k=1$. Then,  choosing $n_{\text{syn}}\gtrsim d^{\gamma}$ suffices to make $D_\alpha(ZV,ZW)$ a constant-factor approximation of the plateau (up to exponentially small terms).
\end{corollary}

A tighter characterization of the regime $n_{\text{syn}} \leq d$ than the work of~\citet{pierquin2025} is an interesting direction for future work, as it seems that some additional privacy gains may be achievable in this regime.

\begin{remark}
    The constants in Proposition~\ref{prop:privacy-gap} can be drastically improved by leveraging the fact that $D_\alpha(V^\top V, W^\top W) \lesssim d^{-\gamma} D_\alpha(V,W)$, as proved in Theorem~\ref{thm:upper-bound-renyi}.
\end{remark}
\subsection{On the generalization to other training procedures, and nonlinear generators}
\label{appsubsec:extensions}

\textbf{Generalization to other training procedures.} The approach developed in this paper may extend to alternative training procedures for linear generators. Importantly, the sufficient-statistics limit  $D_\alpha(ZV,ZW) \overset{n_{\text{syn}} \to \infty}{\longrightarrow} D_\alpha(V^\top V, W^\top W)$ of Proposition~\ref{prop:sufficient-statistics} does not rely on $V$ or $W$ being isotropic Gaussian matrices (related to output perturbation). As noted by~\citet{Brown2024}, Noisy Gradient Descent applied to linear regression converges to a matrix Gaussian distribution, suggesting that our framework could naturally accommodate this more general setting. Extending the analysis in this direction is a promising avenue for future work.\looseness=-1

\textbf{Extension to nonlinear settings.}
Our analysis relies on sufficient-statistics reduction, which makes it possible to analyze the privacy loss of infinite data release through a finite-dimensional object. We believe that there are several promising directions for expanding our analysis to deeper or nonlinear models. 

First, by the same analysis as Proposition~\ref{prop:gram-limit}, it extends directly to deep linear networks, which can capture richer training dynamics: $$D_\alpha(ZV_1\dots V_m, ZW_1\dots W_m) \leq D_\alpha(V_m^\top \dots V_1^\top V_1 \dots V_m, W_m^\top \dots W_1^\top W_1 \dots W_m).$$

Second, by the post-processing theorem, the same reasoning applies to any setting in which only the first layer is fine-tuned.

Third, and most importantly, it can also be extended to deep neural networks through a suitable linearization argument, especially when the input variance is small. In that regime, one may hope to recover our sufficient-statistics approach up to an approximation error, thereby making the analysis relevant in nonlinear settings. There are multiple candidates for this approximation, through Taylor expansion $Z D_z f_V(0)$ or through Hermite expansion $Z {\mathbb{E}}[D_z f_V(Z)|V]$. We give one example below, for which we give heuristics through a control over Wasserstein distances.
Specifically, we can exploit the Gaussian structure of the input noise $Z$ to show that, in some cases, a nonlinear generator $f_v(Z):=f(v,Z)$ can be well-approximated by a linear transformation $ZB$. 

More precisely, conditionally on $V= v$, let $Z \in \mathbb{R}^{n_{\rm syn} \times d}$ have i.i.d rows $Z_i \sim N(0, \sigma^2 I_d)$, and:  $$B_v = {\arg\min}_B \mathbb{E} \|f_v(Z_1) - Z_1B\|_2^2.$$
Assume that $f_v$ is a smooth neural network composed of linear layers with possibly nonlinear odd activations. Then, by integration by parts, we have $Z_1B_v = Z_1\mathbb{E}[D_z f_v(Z_1)]$, which can be seen as the best $L^2$ Gaussian linear approximation of $f_v(Z_1)$. We define the approximation error as $\varepsilon_v^2 = \mathbb{E}\|f_v(Z_1) - Z_1B_v\|^2$. Then, Gaussian Poincaré inequality yields:
$$\varepsilon_v^2 \leq \sigma^2 \mathbb{E} \|D_z f_v(Z_1) - B_v \|^2.$$

Moreover, coupling $f_v(Z)$ and $Z B_v$ with the same latent matrix $Z$ gives, with a renormalisation $1/n_{\rm syn}$:
$$\Omega := W_2(n_{\rm syn}^{-1/2}f_v(Z), n_{\rm syn}^{-1/2}Z B_v)^2 \leq \epsilon_v^2.$$

This gives us an approximation that is uniform in the number of released synthetic records. Consequently, our sufficient statistics reduction becomes applicable again, with the statistics $B_v^\top B_v$, up to a controlled approximation error. Moreover, for randomized parameters $V$, we obtain by the tower property: $\Omega\leq \mathbb{E}[\epsilon_V^2].$

We can give a concrete coarse upper bound for a two-layer scalar neural network with smooth odd activations. Specifically, we assume $\|\phi''\|_\infty \leq L_2, \|\phi'''\|_\infty \leq L_3$. Let $a \in \mathbb{R}^{m}, W \in \mathbb{R}^{m \times d}$, $v = (a,W)$ and, for $z \in \mathbb{R}^d$: $$f_v(z) = \sum_{j=1}^m a_j \phi \left(\langle w_j, z\rangle/\sqrt d\right).$$

We make the same boundedness assumption as in the paper, with a rescaling for $w$ to renormalize the network parameters: $\|a\|_2 \leq C_a, \|w\|_F \leq \sqrt d C_w$, and note $V = (A,W)$, with $A = a + N_1$, $W = w + N_2$. Denoting $U_j = \langle W_j,z\rangle/\sqrt d \sim N(0,\sigma^2 \|W_j\|^2/d)$, we have:
$$\epsilon_V^2 \leq \sigma^2\mathbb{E} \Big[\Big\|\sum_{j=1}^m A_j (\phi'(U_j) - \mathbb{E}[\phi'(U_j)]) W_j\Big\|^2 | V \Big]\leq  \frac{\sigma^2}{d} \sum_{j,l=1}^m A_j A_l (W_j^\top W_l){\rm Cov}(\phi'(U_j),\phi'(U_l) | V).$$

The diagonal terms of the covariance can be upper bounded via the Gaussian Poincaré inequality:

$${\rm Var}(\phi'(U_j) | V) \leq \sigma^2 L_2^2 \|W_j\|^2/d.$$

For the off diagonal terms, we can leverage Price's theorem and $\|\phi'''\|_{\infty} \leq L_3$ to obtain:
$$|{\rm Cov}(\phi'(U_j), \phi'(U_l) | V)| \leq \frac{1}{2} \sigma^4 L_3^2 |W_j^\top W_l|^2/d^2.$$

Averaging over $V$, and leveraging standard moment bounds, we obtain:

$$\Omega \lesssim \sigma^4 L_2^2 (C_A^2+m) + \frac{\sigma^6L_3^2(C_A^2 + m^2)}{2}(d^{-3/2} + (|w_j^\top w_l|/d)^3) $$

The approximation error is thus controlled by the norm of the network, its size and the variance of the Gaussian noise. Therefore, for a simple nonlinear network, the law of the synthetic data release can be made close to that of a linear Gaussian model in Wasserstein distance, with small input variance for example.

Extending this analysis to Rényi divergences is a natural next step. Using the weak triangle inequality for Rényi divergences, we obtain a decomposition that suggests controlling the approximation term $D_{2\alpha}(f(V,Z), Z B_V)$. The analysis developed in this paper already shows that: $$D_\alpha(Z B_V, Z B_W) \to D_\alpha(B_V^\top B_V, B_W^\top B_W).$$
We expect that sharper bounds can be obtained with more refined analysis. Fully extending these results to deep nonlinear generators is an important direction for future work, but the underlying mechanism is already visible in the above simplified setting.
\subsection{Proofs of Proposition~\ref{prop:privacy-gap} and corollary~\ref{corr:privacy-gap-scaling}}
In this section, we derive a lower bound on $D_\alpha(ZV,ZW) - D_\alpha(V^\top V, W^\top W)$ in the case $k=1$. In this setting, we leverage the fact that $D_\alpha(ZV,ZW) = D_\alpha(V^\top Z^\top Z V, W^\top Z^\top Z W) = D_\alpha(\|ZV\|^2,\|ZW\|^2)$.
Here, $ZV \in \mathbb{R}^{n_{\text{syn}}}$. Also, $V^\top V = \|V\|^2$ follows a non-central chi-squared distribution $\chi_d^2(\|v\|^2)$ with degree of freedom $d$ and non-centrality parameter $\|v\|^2$. To simplify notations, we note $n=n_{\text{syn}}$ and: $S_v = \frac{1}{n_{\text{syn}}}\|ZV\|^2$ and $Y_v = \|V\|^2$. We note the measures: $\mu_n = P\left(\frac{1}{n_{\text{syn}}}\|ZV\|^2\right), \nu_n = P\left(\frac{1}{n_{\text{syn}}}\|ZW\|^2\right), \mu = \|V\|^2$ and $ \nu = \|W\|^2$. Also, for $s > 0$, we note $r(x) = \frac{d\mu}{d\nu}(s)$.

\begin{lemma}
    Assume that $k=1$. Let $\alpha > 1$. For two distributions $P, Q$, we note $H_\alpha(P,Q) = \exp((\alpha-1)D_\alpha(P,Q))$. Let $X_n \sim \frac{1}{n}\chi_{n}^2$ follow a (rescaled) chi-squared distribution with degree of freedom $n$. Then,
    \[H_\alpha(\mu, \nu) - H_\alpha(\mu_n, \nu_n) = \mathbb{E}_{S_w\sim \nu_n}\left[\mathbb{E}_{Y_w \sim \nu}[r(Y_w)^\alpha | Y_wX_n = S_w] - \mathbb{E}_{Y_w \sim \nu}\left[r(Y_w) \big| Y_wX_n = S_w \right]^\alpha\right].\]
\end{lemma}

\begin{proof}
    First, we notice that  $ZV | V \sim \mathcal{N}(0,\|V\|^2 I_{n})$. Let $X_n \sim \frac{1}{n}\chi_{n}^2$. Then, $S_v | Y_v = Y_v X_n$. $\mu_n$ and $\nu_n$ can be seen as the output of the same Markov kernel $K : y \mapsto yX_n$ for the measures $\mu$ and $\nu$. Thus, the density ratio $\frac{d\mu_{n}}{d\nu_{n}}(s)$ writes:
    \[\frac{d\mu_{n}}{d\nu_{n}}(s) = \mathbb{E}_{Y_w \sim\nu}\left[\frac{d\mu}{d\nu}(Y_w) \big| Y_wX_n = s \right] = \mathbb{E}_{Y_w \sim \nu}\left[r(Y_w) \big| Y_wX_n = s \right]\]
    Now, by definition,
    \[H_\alpha(S_v,S_w) = \mathbb{E}_{S_w \sim \nu_n}\left[\mathbb{E}_{Y_w \sim \nu}\left[r(Y_w) \big| Y_wX_n = S_w \right]^\alpha\right]\]
    
    Also by tower property, \[H_\alpha(Y_v,Y_w) = \mathbb{E}_{Y_w \sim \nu}\left[r(Y_w)^\alpha\right] = \mathbb{E}_{S_w \sim \nu_n}\left[\mathbb{E}_{Y_w \sim \nu}\left[r(Y_w)^\alpha \big| Y_wX_n = S_w \right]\right],\]
    giving the desired result.
    \end{proof}

    Now, we upper bound the quantity $H_\alpha(\mu, \nu) - H_\alpha(\mu_n, \nu_n)$. Our strategy relies on showing that, for large values of $n$, the ratio $r(Y_w)$ is concentrated around $\mathbb{E}[r(Y_w) | Y_w X_n = S_w]$ with high probability.

    \begin{lemma}
    \label{lem:divergence-decomposition-moments}
    Let $\alpha \geq 2$. We note $M(S_w) = \mathbb{E}[r(Y_w) | Y_w X_n = S_w]$. Then,
    
        \[H_\alpha(\mu, \nu) - H_\alpha(\mu_n, \nu_n) \leq \alpha(\alpha-1) 2^{\alpha-3}e^{(\alpha-1)(\alpha-2)\Delta^2/2} \|r(Y_w)-M(S_w)\|_\alpha^2.\]
    \end{lemma}
    
    \begin{proof}
    We have, by definition:
    \[\mathbb{E}[r(Y_w) - M(S_w)|Y_w X_n = S_w] = \mathbb{E}[r(Y_w) - M(S_w)|S_w]  =  0.\]
    Then, $\mathbb{E}[\alpha M^{\alpha-1}(S_w)(r(Y_w)-M(S_w))|S_w] = 0$. Let $\phi_\alpha: x \mapsto x^\alpha$ Then, leveraging $\alpha \geq 2$ and by the mean value theorem in the expectation in the interval $(\min(M(S_w),r(Y_w)),\max(M(S_w),r(Y_w))$, and $\max(M(S_w),r(Y_w)) \leq M(S_w) + r(Y_w)$:
    \begin{align*}
         H_\alpha(\mu, \nu) - H_\alpha(\mu_n, \nu_n) &= \mathbb{E}\left[\mathbb{E}[r(Y_w)^\alpha |S_w] -M(S_w)^\alpha\right]\\
        &= \mathbb{E}\left[\mathbb{E}[r^\alpha(Y_w) - M(S_w)^\alpha - \alpha M^{\alpha-1}(S_w)(r(Y_w)-M(S_w))|S_w]\right]\\
        &= \mathbb{E}\left[\mathbb{E}[\phi(r(Y_w)) - \phi(M(S_w)) -  (r(Y_w)-M(S_w)) \phi'(M(S_w))|S_w]\right]\\
        &=\frac{1}{2} \mathbb{E}\left[\mathbb{E}[\phi''(Z) (M(S_w)-r(Y_w))^2|S_w]\right], \;\;\;\; Z \in (\min(M(S_w),r(Y_w)),\max(M(S_w),r(Y_w))\\
        &\leq \frac{\alpha(\alpha-1)}{2} \mathbb{E}\left[\mathbb{E}[(M(S_w)+r(Y_w))^{\alpha-2} (M(S_w)-r(Y_w))^2|S_w]\right]\\
        &\leq \frac{\alpha(\alpha-1)}{2} \|r(Y_w)+M(S_w)\|_\alpha^{\alpha-2} \|r(Y_w)-M(S_w)\|_\alpha^2,
    \end{align*}
    where the last step is the Hölder inequality with exponents $\frac{\alpha}{\alpha-2}$ and $\frac{\alpha}{2}$.
    Now, by tower property and the Jensen inequality,
    \[\|M(S_w)\|_\alpha^\alpha  = \mathbb{E}[M(S_w)^\alpha] = \mathbb{E} [\mathbb{E}[r(Y_w)|S_w]^\alpha] \leq \mathbb{E} [\mathbb{E}[r(Y_w)^\alpha|S_w]] = \|r(Y_w)\|_\alpha^\alpha.\]
    Then, triangle inequality implies: 
    \[\|r(Y_w) + M(S_w)\|_\alpha \leq \|r(Y_w)\|_\alpha + \|M(S_w)\|_\alpha \leq 2 \|r(Y_w)\|_\alpha.\]
    Therefore,

    \[H_\alpha(\mu, \nu) - H_\alpha(\mu_n, \nu_n) \leq \alpha(\alpha-1) 2^{\alpha-3}\|r(Y_w)\|_\alpha^{\alpha-2} \|r(Y_w)-M(S_w)\|_\alpha^2.\]

    Also, we can relate $\|r(Y_w)\|_\alpha$ to $D_\alpha(V,W)$ by a post-processing argument: \[\|r(Y_w)\|_\alpha^\alpha = \mathbb{E}\left[r(Y_w)^\alpha\right] = e^{(\alpha-1)D_\alpha(\|V\|^2,\|W\|^2)} \leq e^{(\alpha-1)D_\alpha(V,W)} = e^{\alpha(\alpha-1)\Delta^2/2} .\]

    Then, we obtain: 
    \[H_\alpha(\mu, \nu) - H_\alpha(\mu_n, \nu_n) \leq \frac{\alpha(\alpha-1)}{2} 2^{\alpha-2}e^{(\alpha-1)(\alpha-2)\Delta^2/2} \|r(Y_w)-M(S_w)\|_\alpha^2.\]
    \end{proof}
    Now, we upper bound $\|r(Y_w)-M(S_w)\|_\alpha^2$ by a concentration argument:
    \begin{proposition}
    Let $\alpha \geq 2$. Then, there exists $C > 0$ such that $d > C$ implies:
        \[H_\alpha(\mu, \nu) - H_\alpha(\mu_n, \nu_n) = O(1/n),\]
        where the rate does not depend on $d$.
    \end{proposition}

    \begin{proof}
    First, for any measurable function $g$, $\|r(Y_w) - M(S_w)\|_\alpha \leq 2\|r(Y_w) - g(S_w)\|_\alpha$.
    In fact, $r(Y_w) - M(S_w) = r(Y_w) - g(S_w) - \mathbb{E}[r(Y_w) - g(S_w) | S_w]$ and the conditional expectation is a contraction of $L^\alpha(\mathbb{R)}$.
    
    Let $A = \{X_n \in (1/2, 2)\}$. First, we upper bound $P(A^C)$ by concentration. $X_n \sim \frac{1}{n}\chi_n^2$. The moment-generating function of $X_n$ is finite for $t \in (-\infty, n/2)$ and writes $M_{X_n}(t) = \left(1- \frac{2t}{n}\right)^{-n/2}$. Let $t\in (0,n/2)$. Then, the Chernoff inequality gives:
    \[P(X_n \geq 2) \leq M_{X_n}(t) e^{-2t}  =\exp\left(-2t - \frac{n}{2} \log(1-2t/n) \right).\]

    We find the lower bound of this inequality at $t^* =  n(1-\frac{1}{x})/2$.
    
    $P(X_n \geq 2) \leq e^{-\frac{n}{2}(1- \log 2)}$.
    Now, let $t < 0$. Since $x \mapsto e^{tx}$ is decreasing:
    \[P(X_n \leq 1/2) = P(e^{Xt} \geq e^{t/2})  \leq M_{X_n}(t) e^{-t/2} = \exp\left(-t/2 - \frac{n}{2} \log(1-2t/n) \right).\]
    We find the lower bound of this inequality at $t^* = -n/2$. Then, $P(X_n \leq 1/2) \leq e^{-\frac{n}{2}(\log 2 - 1/2)}$
    Also, $\log 2 - 1/2 \approx 0.19 \leq 1 - \log 2 \approx 0.31$.
    Therefore, $P(A^C) \leq 2 e^{-cn}$, with $c = \frac{1}{2}(\log 2 - 1/2)$.
    
     Then, let $\kappa \in (0, c/\alpha)$. Choosing $g_n(S_w) = r(S_w) \wedge e^{\kappa n}$, we have:
    \begin{align*}
        \|r(Y_w)-g(S_w)\|_\alpha^\alpha = \mathbb{E}[|r(Y_w) - g_n(S_w) \mathbbm{1}_A|^\alpha] + \mathbb{E}[|r(Y_w) - g_n(S_w) \mathbbm{1}_{A^C}|^\alpha]. 
    \end{align*}

    We note $T_A :=  \mathbb{E}[|r(Y_w) - g_n(S_w)|^\alpha \mathbbm{1}_A]$ and $T_{A^C} := \mathbb{E}[|r(Y_w) - g_n(S_w)|^\alpha \mathbbm{1}_{A^C}]$
    On $A$, \[|r(Y_w) - g_n(S_w)| \leq |r(Y_w) - r(S_w)| + (r(S_w) - e^{\kappa n})_+.\]

    Then, \[T_A \leq 2^{\alpha-1} T_1 + 2^{\alpha-1} T_2,\]
    with $T_1 := \mathbb{E}[|r(Y_w) - r(S_w)|^\alpha \mathbbm{1}_A]$ and $T_2 = \mathbb{E}[(r(S_w) - e^{\kappa n})_+^\alpha \mathbbm{1}_A]$.

    We first bound $T_2$. Let $\beta > \alpha$. Since $(u-e^{\kappa n})_+^\alpha \leq u^\beta e^{-(\beta-\alpha)\kappa n}$, we have:

    \[\mathbb{E}[(r(S_w) - e^{\kappa n})_+^\alpha \mathbbm{1}_A] \leq e^{-(\beta - \alpha) \kappa n} \mathbb{E}[r(S_w)^\beta \mathbbm{1}_A].\]

    We show that $\mathbb{E}[r(S_w)^\beta \mathbbm{1}_A]$ is uniformly bounded for sufficiently large $d$. We note $\lambda_v = \|v\|$, $\lambda_w =  \|w\|$ and $B = (\lambda_v^2 + \lambda_w^2)/2$. Let $y > 0$. The density of $Y_w$ writes:

    \begin{align*}
        p_{\lambda_v}(y) = p_0(y) e^{- \lambda_v^2/2} \lambda_v^{1-d/2}  I_{d/2-1}\big(\lambda_v\sqrt{y}\big).
    \end{align*}

    Then, we write the ratio $r(y)$:
    \begin{align*}
        r(y) = e^{(\lambda_w^2-\lambda_v^2)/2} \left(\frac{\lambda_v}{\lambda_w}\right)^{1-d/2}\frac{I_{d/2-1}\big(\lambda_v\sqrt{y}\big)}{I_{d/2-1}\big(\lambda_w\sqrt{y}\big)}.
    \end{align*}
    We recall the following relationship:
    \[\lambda_v \sqrt{y}R_{d/2-1}(\lambda_v \sqrt{y}) + d/2 - 1= \lambda_v \sqrt{y} \frac{I_{d/2-1}'(\lambda_v \sqrt{y})}{I_{d/2-1}(\lambda_v \sqrt{y})}.\]

    Then, we compute an upper bound to $\frac{d\log r}{dy}(y)$.
    \begin{align*}
        \frac{d\log r}{dy}(y) &= \partial_y \log I_{d/2-1}(\lambda_v\sqrt{y})- \partial_y \log I_{d/2-1}(\lambda_w\sqrt{y})\\
        &= \frac{1}{2\sqrt{y}}\left(\lambda_v\frac{I_{d/2-1}'(\lambda_v\sqrt{y})}{I_{d/2-1}(\lambda_v\sqrt{y})} - \lambda_w\frac{I_{d/2-1}'(\lambda_w \sqrt{y})}{I_{d/2-1}(\lambda_w \sqrt{y})}\right)\\
        &= \frac{1}{2\sqrt{y}}\left(\lambda_vR_{d/2-1}(\lambda_v \sqrt{y}) - \lambda_wR_{d/2-1}(\lambda_w \sqrt{y})\right)\\
        &\leq \frac{1}{2\sqrt{y}}\left(\lambda_vR_{d/2-1}(\lambda_v \sqrt{y}) + \lambda_wR_{d/2-1}(\lambda_w \sqrt{y})\right)\\
        &\leq \frac{\lambda_v^2+ \lambda_w^2}{2d},
    \end{align*}
    where the Bessel quotient upper bound is obtained leveraging the following inequality:
    \[ 0 \leq R_{\nu}(y) =\frac{I_{\nu+1}(y)}{I
    _{\nu}(y)} \leq \frac{y}{2(\nu+1)}.\]

    Therefore, $|r'(y)| = r(y) |\partial_y \log r(y) | \leq \frac{B}{d} r(y)$. Then, by Gronwall's inequality, for any $x \in (1/2,2)$, 
    \[r(yx) \leq r(y) e^{B y |x-1|/d}\leq r(y)e^{By/d}.\]

    Then, by Cauchy-Schwarz inequality:
    
    \[\mathbb{E}\left[r(S_w)^\beta \mathbbm{1}_A\right] = \mathbb{E}\left[r(Y_w X_n)^\beta \mathbbm{1}_A\right] \leq \mathbb{E}\left[r(Y_w) e^{2\beta B Y_w /d}\right] \leq \mathbb{E}\left[r(Y_w)^{2\beta}\right]^{1/2} \mathbb{E}\left[e^{2\beta B Y_w/d}\right]^{1/2}.\]

    We can upper bound the first factor by post-processing:

    \[\mathbb{E}\left[r(Y_w)^{2\beta}\right] = e^{(2\beta-1)D_{2\beta}(\mu,\nu)} \leq e^{(2\beta
    -1)D_{2\beta}(V,W)} =e^{\beta(2\beta-1)\Delta^2}.\]
    Also, leveraging the moment-generating function of non-central chi-squared distributions, for $t < 1/2$:
    \[\mathbb{E}[e^{t Y_w}] = (1-2t)^{-d/2} \exp(\lambda_w^2 t/(1-2t)).\]

    Taking $t = 2\beta B / d$, we get a bound independent of $d$ for all sufficiently large $d$. Therefore,
\[T_2 \leq C_{\alpha,\beta, v,w} e^{-(\beta- \alpha) \kappa n}.\]

We now bound $T_1$.
    

    Noticing $(r(Y_w) - r(Y_wX_n))^\alpha \mathbbm{1}_{A} \leq |Y_w - Y_wX_n|^\alpha\sup_{t \in (Y_w/2, 2Y_w)} r'(t)^\alpha = Y_w^\alpha|1 - X_n|^\alpha\sup_{x \in (1/2, 2)} r'(Y_w x)^\alpha$, we get:
    \begin{align*}
        \mathbb{E}[(r(Y_w) - r(S_w))^\alpha\mathbbm{1}_{A}]
     &\leq \mathbb{E}[Y_w^\alpha |1-X_n|^\alpha \sup_{x \in (1/2,2)} r'(Y_wx)^\alpha]\\
        &\leq\mathbb{E}[|1-X_n|^\alpha] \mathbb{E}[Y_w^\alpha  \sup_{x \in (1/2,2)} r'(Y_w x)^\alpha]\\
    \end{align*}

     Now, we upper bound the quantity $\mathbb{E}[|1-X_n|^\alpha]$. $nX_n-n$ is a sub-Gamma random variable, since $n X_n - n \in \Gamma(2n,2)$ of~\cite{Boucheron2013}. From Theorem 2.3. of~\citet{Boucheron2013}, for even $\alpha > 1$, $\mathbb{E}[|n X_n - n|^\alpha] \leq (\alpha/2)! (16n)^{\alpha/2} + \alpha!8^\alpha \lesssim C_\alpha n^{\alpha/2}$.
    Dividing by $n^\alpha$, there exists $C_\alpha > 0$ such that $\mathbb{E}[|1-X_n|^\alpha] \leq C_\alpha n^{-\alpha/2}$. For other values of $\alpha$, we use the monotonicity of the $L_p$ norm to obtain the upper bound.

    Then, we obtain \[\mathbb{E}[(r(Y_w) - r(S_w))^\alpha\mathbbm{1}_{A}] \leq C_\alpha n^{-\alpha/2}\mathbb{E}[Y_w^\alpha \sup_{x \in (1/2,2)} r'(Y_w x)^\alpha].\]

    Now, we evaluate $\mathbb{E}[Y_w^\alpha \sup_{x \in (1/2,2)} r'(Y_w x)^\alpha]$. Let $y > 0$. We have: \[\sup_{x \in (1/2,2)} |r'(Y_w x)| \leq \frac{B}{d} \sup_{x \in (1/2,2)} |r(Y_w x)| \leq \frac{B}{d} r(Y_w) e^{2yB/d}.\]

    
    Then, we apply the generalized Hölder inequality:
    \begin{align*}
        \mathbb{E}[Y_w^\alpha \sup_{x \in (1/2,2)} r'(Y_wx)^\alpha] &\leq \frac{B^\alpha}{d^\alpha}  \mathbb{E}[Y_w^\alpha  r(Y_w)^\alpha e^{2\alpha Y_wB/d}]\\
        &\leq \frac{B^\alpha}{d^\alpha}  \mathbb{E}[Y_w^{3\alpha}]^{1/3} \mathbb{E}[e^{6\alpha Y_w B/d}]^{1/3} \mathbb{E}[r(Y_w)^{3\alpha}]^{1/3}\\
        &\leq \frac{B^\alpha}{d^\alpha}\mathbb{E}[Y_w^{3\alpha}]^{1/3} \mathbb{E}[e^{6\alpha Y_w B/d}]^{1/3} e^{(3\alpha-1)D_{3\alpha}(V,W)/3}\\
        &\leq \frac{B^\alpha}{d^\alpha} \mathbb{E}[Y_w^{3\alpha}]^{1/3} \mathbb{E}[e^{6\alpha Y_w B/d}]^{1/3}  e^{\alpha(3\alpha-1)\Delta^2/2}
    \end{align*}
    Also, leveraging the moment-generating function of non-central chi-squared distributions, for $d > 12 \alpha  B$:
    \[\mathbb{E}[e^{6\alpha Y_w B/d}] = (1-12\alpha B/d)^{-d/2} \exp(6\alpha B\lambda_w^2/(d-12\alpha B)).\]
    Also, leveraging Lemma~\ref{lem:chi-2-moments}, we can upper bound $\mathbb{E}[Y_w^{3\alpha}]$:
    \[\mathbb{E}[Y_w^{3\alpha}] \leq (d+ \alpha C_1 (2 + \lambda_w)/2)^{3\alpha},\]
    with $C_1 > 0$.
    Then, we obtain:
    \begin{align*}
        \mathbb{E}[Y_w^\alpha \sup_{x \in (1/\beta,\beta)} r'(Y_wx)^\alpha] &\leq \frac{B^\alpha}{d^\alpha} \mathbb{E}[Y_w^{3\alpha}]^{1/3} \mathbb{E}[e^{6\alpha Y_w B/d}]^{1/3}  e^{\alpha(3\alpha-1)\Delta^2/2}\\
        &\leq B^\alpha \left(1+ \frac{\alpha C_1 (2 + \lambda_w^2)}{2d}\right)^{\alpha}  \left(1- \frac{12\alpha B}{d}\right)^{-d/6} e^{\alpha(3\alpha-1)\Delta^2/2 + 2\alpha B\lambda_w^2/3(d-12\alpha B)}.\\
        &\lesssim C'_{v,w,\alpha} \text{, independent of $d$ for sufficiently large $d$}.
    \end{align*}
    In particular, this upper bound does not diverge when $d \to \infty$. This means that for $d$ large enough, we get convergence rates that do not depend on the dimensionality of the model $d$. However, our proof technique generates a factor $\exp(6\alpha B\lambda_w/(d-12\alpha B))$ appear, which means that our analysis only holds for large values of $d$.

    Therefore, $T_1 \leq C'_{\alpha, v,w} n^{-\alpha/2}$.

    Finally, we control $T_{A^C}$. Since $0 \leq g_n(S_w) \leq e^{\kappa n}$, $T_{A^C} \leq 2^{\alpha-1} \mathbb{E}\left[r(Y_w)^\alpha \mathbbm{1}_{A^C}\right] + 2^{\alpha-1} e^{\alpha \kappa n} P(A^C)$.

    By post-processing and Cauchy-Schwarz, $\mathbb{E}\left[r(Y_w)^\alpha \mathbbm{1}_{A^C}\right] \leq C_{\alpha,v,w,} e^{-cn/2}$. Moreover, $e^{\alpha \kappa n} P(A^C) \leq 2 e^{-(c - \alpha \kappa) n}$ with $\kappa < c/\alpha$. Therefore, $T_{A^C} \leq C''_{\alpha, t} e^{-c' n}$ for some $c' > 0$.

    Then, combining those bounds, there exists $C,C', c'' > 0$ such that    
    \[H_\alpha(\mu, \nu) - H_\alpha(\mu_n, \nu_n) \leq \frac{C}{n} + C' e^{-c'' n}.\]
\end{proof}

\begin{lemma}
\label{lem:chi-2-moments}
    Let $d > 1, \theta > 0$ and $Y_w \sim \chi_d^2(\theta^2)$ follow a non-central chi-squared distribution. Let $\alpha > 1$. Then, there exists $C_0 > 0$ such that
    \[\|Y\|_\alpha \leq 2(d/2 + \alpha C_1(2+ \theta^2)).\]
\end{lemma}

\begin{proof}
    Let $K \sim \text{Poisson}(\theta^2/2)$. Then, non-central chi-squared distributions admit the following Poisson mixture of chi-squared distributions representation: $Y | K \sim\chi_{d+2K}^2$. Also, central chi-squared distributions $X_d \sim \chi_d^2$ admit the following moments: $E[X_d^\alpha] = 2^\alpha \frac{\Gamma(\alpha + d/2)}{\Gamma(d/2)}$, where $\Gamma$ denotes the $\Gamma$ function.
    Therefore, leveraging the convexity inequality $(x+y)^\alpha \leq 2^{\alpha-1} (x^\alpha + y^\alpha)$:
    \begin{align*}
        \mathbb{E}[Y^\alpha] = \mathbb{E}[\mathbb{E}[Y^\alpha | K]] = \mathbb{E}\left[\mathbb{E}[X_{d + 2K}^\alpha|K]\right] = 2^\alpha\mathbb{E}\left[\frac{\Gamma(\alpha + K + d/2)}{\Gamma(K + d/2)}\right] &\leq 2^\alpha\mathbb{E}\left[(\alpha + K + d/2)^\alpha\right].
    \end{align*}
    Then, by triangle inequality: $\|Y\|_\alpha \leq 2 \| \alpha + d/2 + K\|_\alpha \leq 2(\alpha + d/2 + \|K\|_\alpha)$.
    
    The moment generating function of $K$ writes $\mathbb{E}[e^{tK}] = \exp(\theta^2 (e^t-1)/2)$. Since $K$ is subexponential, we can simply upper bound $\|K\|_\alpha $ with Orlicz norms~\citep{Vershynin2018}: $\|K\|_\alpha \leq C_0 \alpha \|K\|_{\psi_1}$, with $C_0 > 0$ an absolute constant. Since $E[K] = \theta^2/2$, we can upper bound the Orlicz norm by $\|K\|_{\psi_1} \lesssim 1+\theta^2/2$.  Therefore, there exists $C_1 > 0$ such that $\|K\|_\alpha \leq C_1 \alpha (1+\theta^2/2)$
    Then,
    \[\|Y\|_\alpha \leq 2(d/2 + \alpha (1+ C_1(1+ \theta^2/2)).\]
\end{proof}

\newpage

\section{Regularity of the studied distributions}
\label{app:regularity}

In this paper, we study several distributions and parametrizations. Our results hold under regularity conditions, imposed by Proposition~\ref{prop:Fisher-upper-bound} and Proposition~\ref{prop:app-renyi-criterion}. In this section, we prove that the studied distributions are sufficiently regular, as defined by Definition~\ref{defi:regularity}.

\begin{proposition}[Regularity of non-central Wishart distributions]
    Let $d, k \geq 1$, and $v,w \in \mathbb{R}^{d \times k}$ We consider the family of distributions:
    \[P := \{P_\theta := \mathcal{W}(k,d, (\theta v + (1-\theta)w)^\top (\theta v + (1-\theta)w)); \theta \in \mathbb{R}\}.\]
    Then, $P$ is sufficiently regular, as defined by Definition~\ref{defi:regularity}.
\end{proposition}
\begin{proof}
    First, $P_\theta$ is the pushforward the distribution of a Gaussian matrix $X_\theta = N + \theta v + (1-\theta) w$ with distribution $Q_\theta$ by the map $T: x \mapsto x^\top x$. We note $Y_\theta \sim P_\theta$. Every $Q_\theta$ with $\theta \in \mathbb{R}$ are mutually absolutely continuous on $\mathbb{R}^{d \times k}$. Therefore, their pushforward through $T$ are mutually absolutely continuous. Therefore, we can define the likelihood ratio $\frac{dP_{\theta'}}{dP_\theta}$ for every $\theta',\theta \in \mathbb{R}$. Let $\theta', \theta \in \mathbb{R}$.  Then,
    \[\frac{dP_{\theta'}}{dP_\theta}(Y_\theta) = \mathbb{E}\left[\exp\left((\theta' - \theta)\langle N, v-w\rangle_F - \frac{(\theta'-\theta)^2}{2} \|v-w\|_F^2\right) \Bigg| Y_\theta\right].\]

    Since the integrand is $C^\infty$ in $L^q$ for every $q \geq 1$ and the conditional expectation is a contraction in $L^q(Q_\theta)$, $\frac{dP_{\theta'}}{dP_\theta}(T(X_\theta))$ is also $C^\infty$ in $L^q(Q_\theta)$, and therefore $\frac{dP_{\theta'}}{dP_\theta}$ is $C^\infty$ in $L^q(P_\theta)$.
\end{proof}


\end{document}